\documentclass[3p,twocolumn]{elsarticle}
\usepackage{amscd,amssymb,graphics}
\usepackage{mathrsfs} 
\usepackage{epsfig}
\usepackage{amsthm}

\newtheorem{theorem}{Theorem}[section]
\newtheorem{corollary}[theorem]{Corollary} 

\newtheorem{proposition}[theorem]{Proposition}

\newtheorem{conjecture}[theorem]{Conjecture}

\newcommand{\abs}[1]{\vert#1\vert}
\def\norm#1{\left\Vert#1\right\Vert}

\def\I {{\mathbb I}}

\def\E {{\mathbb E}}
\def\N{{\mathbb N}}

\def\R{{\mathbb R}}
\def\s{{\mathbb S}}
\def\e{\varepsilon}
\def\ve{\varepsilon}
\def\VC{{\mathrm{VC}}}  
\def\Ur{{\mathbb U}}
\def\var{{\mathrm{var}}\,}

\def\diam{{\mathrm{diam}}\,}

\journal{Journal of Discrete Algorithms}

\begin{document}

\begin{frontmatter}

%% Title, authors and addresses

%% use the tnoteref command within \title for footnotes;
%% use the tnotetext command for the associated footnote;
%% use the fnref command within \author or \address for footnotes;
%% use the fntext command for the associated footnote;
%% use the corref command within \author for corresponding author footnotes;
%% use the cortext command for the associated footnote;
%% use the ead command for the email address,
%% and the form \ead[url] for the home page:
%%
%% \title{Title\tnoteref{label1}}
%% \tnotetext[label1]{}
%% \author{Name\corref{cor1}\fnref{label2}}
%% \ead{email address}
%% \ead[url]{home page}
%% \fntext[label2]{}
%% \cortext[cor1]{}
%% \address{Address\fnref{label3}}
%% \fntext[label3]{}

\title{Indexability, concentration, and VC theory}

%% use optional labels to link authors explicitly to addresses:
%% \author[label1,label2]{<author name>}
%% \address[label1]{<address>}
%% \address[label2]{<address>}

\author{Vladimir Pestov}

\address{Departamento de Matem\'atica,
Universidade Federal de Santa Catarina, Campus Universit\'ario Trindade,
CEP 88.040-900 Florian\'opolis-SC, Brasil \footnote{}\fnref{brasil}}
\fntext[brasil]{CNPq visiting researcher} 
\address{Department of Mathematics and Statistics, University of Ottawa,
        585 King Edward Avenue, Ottawa, Ontario K1N6N5 Canada  \footnote{}\fnref{canada}}
\fntext[canada]{Permanent address}

\begin{abstract}
Degrading performance of indexing schemes for exact similarity search in high dimensions has long since been linked to histograms of distributions of distances and other $1$-Lipschitz functions getting concentrated. We discuss this observation in the framework of the phenomenon of concentration of measure on the structures of high dimension and the Vapnik-Chervonenkis theory of statistical learning. 
\end{abstract}

\begin{keyword}
%% keywords here, in the form: keyword \sep keyword
Exact similarity search \sep 
indexing schemes \sep
curse of dimensionality \sep
Lipschitz functions \sep
concentration of measure\sep
uniform Glivenko--Cantelli theorem \sep
pivot tables \sep
metric trees 
%% MSC codes here, in the form: \MSC code \sep code
%% or \MSC[2008] code \sep code (2000 is the default)
\MSC[2010] 68P10 \sep 68P20 \sep 68Q87
\end{keyword}

\end{frontmatter}

%%
%% Start line numbering here if you want
%%
% \linenumbers

%% main text
\section{Introduction}
At an intuitive level, at least for a limited class of indexing schemes the geometric and probabilistic origin of the
curse of dimensionality is quite transparent.
Let $W=(\Omega,\rho,X)$ denote a similarity workload, where $\rho$ is a metric on a domain $\Omega$ and $X$ is a finite subset of $\Omega$ (dataset). Let us say we are interested in indexing  into $W$ for deterministic, exact range queries. A traditional ``distance-based'' indexing scheme, stripped down to the bone, consists of a family of real-valued functions $f_i$, $i\in I$ on $\Omega$, either fully or partially defined, which satisfy the $1$-Lipschitz property:
\begin{equation}
\abs{f_i(x)-f_i(y)}\leq \rho(x,y).\end{equation}
(For example, a pivot-based indexing scheme will be using distance functions $\rho(p_i,-)$ to the pivots $p_i\in\Omega$.)
Given a range query $(q,\e)$, where $q\in\Omega$ and $\e>0$, the algorithm chooses recursively a sequence of indices $i_n$, where each $i_{n+1}$ is determined by the values $f_{i_k}(q)$, $k\leq n$.
The functions $f_i$ serve to discard those datapoints which cannot possibly answer the query. Namely, if $\abs{f_i(q)-f_i(x)}\geq\e$, then, by the $1$-Lipschitz property of $f_i$, one has $\rho(q,x)\geq \e$, and so the point $x$ is irrelevant and need not be considered (Figure \ref{fig:discarding}).

\begin{figure}[ht]
\begin{center}
\scalebox{0.25}[0.25]{\includegraphics{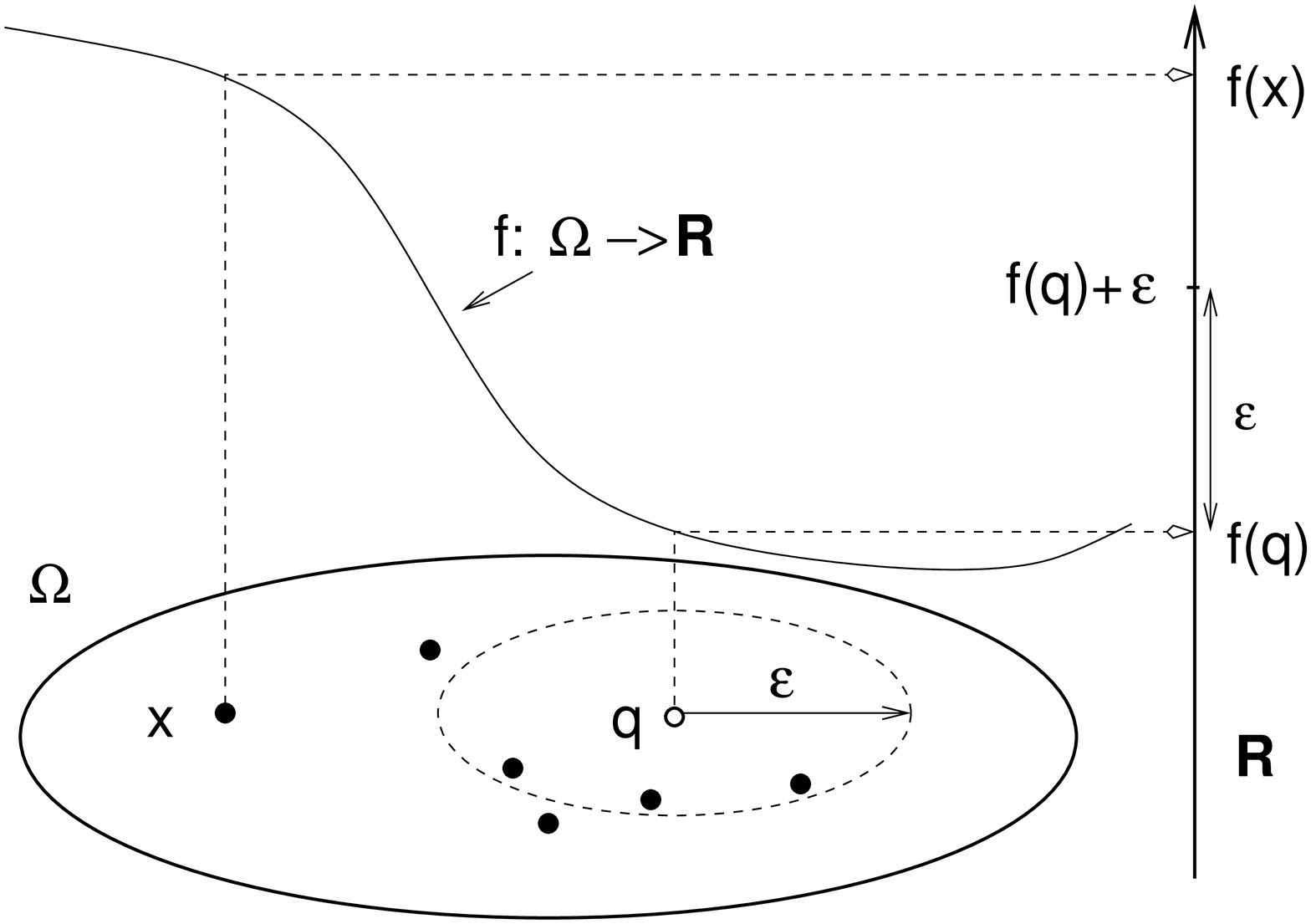}} 
\caption{The datapoint $x$ can be discarded.
\label{fig:discarding}}
\end{center}
\end{figure}  

After the calculation terminates,
the algorithm returns all points which cannot be discarded, and checks each one of them against the condition $\rho(x,q)<\e$.

Next come two standard observations about high-dimensional data. The first one, known as the ``empty space paradox,'' asserts that the average distance $\E(\e_{NN})$ to the nearest neighbour approaches the average distance $\E(\rho)$ between two datapoints as the dimension $d$ goes to infinity, provided the number of datapoints, $n$, grows subexponentially in $d$.
Cf. Figure \ref{fig:nn_dist}, where we illustrate the point with a constant number of points ($n=10^3$ and $n=10^5$), and the distances are normalized so that the {\em characteristic size} of the gaussian space $(\R^n,\gamma^n)$,
\begin{equation}
{\mathrm{CharSize}}\,(X) =\E_{\mu\otimes\mu}(\rho(x,y)),\end{equation}
is one.

\begin{figure}[ht]
\centering
\epsfig{file=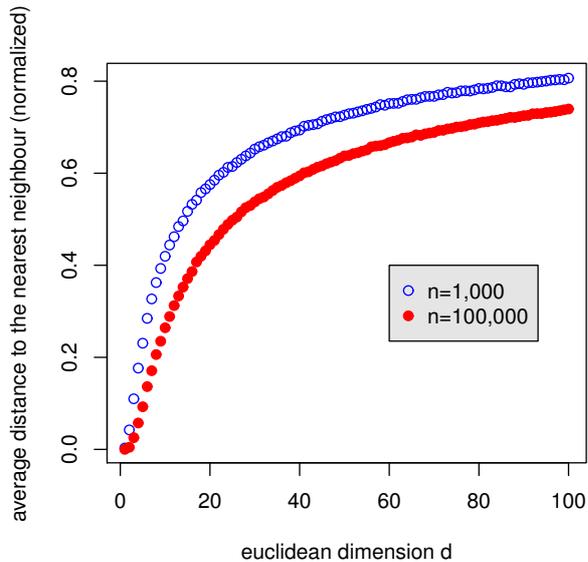, height=3.3in, width=3.3in} 
\caption{\label{fig:nn_dist}
The normalized average distance to the nearest neighbour in a dataset of $n$ points randomly drawn from a gaussian distribution in $\R^d$.}
\end{figure}

The second observation is that the histograms of values of common $1$-Lipschitz functions on high-dimensional data are concentrated near their mean (or median) values. This effect is already pronounced in moderate dimensions such as $d=14$ in Figure \ref{fig:hist_14gauss}. Here the function is a distance to a randomly chosen pivot $p$, and assuming the query point $q$ is at a distance $\approx 1$ from $p$, only the points outside of the region marked by vertical bars can be discarded.

% n =100000
% d = 14
% X <- matrix(rnorm(n*d, mean=0, sd=1), ncol=d)
% sample <- sample(n,100)
% S <- X[sample,]
% D <- pairwiseDistances(S,X)
% ED <- mean(D)
% DD <- D[1,]/ED

% hist(DD,col="5",xlab="normalized distance",breaks=30,main="",xlim=c(0,2))
% postscript(file="hist_14gauss.eps",width=5,height=5,horizontal=FALSE,onefile=FALSE)
% hist(DD,col="5",xlab="normalized distance",breaks=30,main="",xlim=c(0,2))
% abline(v=(1+rad),col="red",lwd=5)
% abline(v=(1-rad),col="red",lwd=5)
% abline(v=1,col="red",lwd=5,lty=3)
% dev.off()

\begin{figure}[ht]
\centering
\epsfig{file=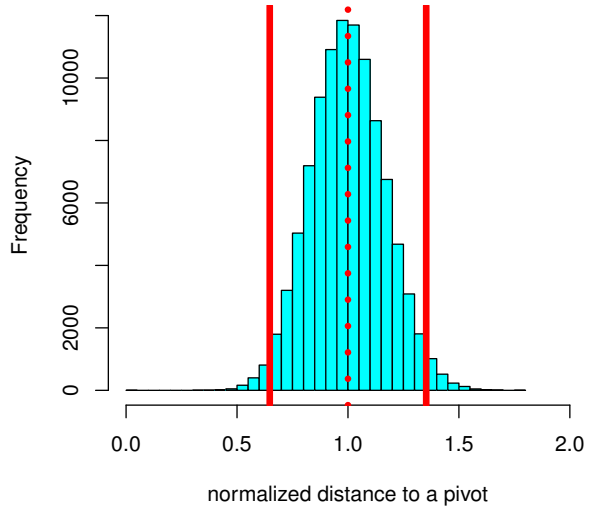, height=3.3in, width=3.3in} 
\caption{\label{fig:hist_14gauss}
Histogram of distances to a randomly chosen pivot in a dataset $X$ of $n=10^5$ points drawn from a gaussian distribution in $\R^{14}$. The vertical lines mark the mean normalized distance $1\pm \e_{NN}$.}
\end{figure}

The two properties combined imply that as $d\to\infty$, fewer and fewer datapoints can be discarded for an average range query, and the performance of an indexing scheme degrades rapidly. This mechanism has been discussed repeatedly, e.g. \cite{chavez:01}, pp. 35--37, \cite{navarro:2002}, \cite{samet}, p. 487, to mention just a few sources.

To make this idea yield rigorous lower performance bounds, one needs to guarantee first that {\em every} histogram of distances of $1$-Lipschitz functions used to build an indexing scheme for a given domain $\Omega$ is highly concentrated. In other words, if $\mathscr F$ denotes a class of $1$-Lipschitz functions from which we can choose the $f_i$, then we want a low uniform upper bound on the variances of $f\in {\mathscr F}$. Results of this type are indeed well-known for a variety of geometric objects and are referred jointly as the {\em phenomenon of concentration of measure} \cite{GrM:83,MS,ledoux}. 

Next problem is, how to link the concentration of functions $f$ with regard to the presumed {\em underlying distribution} $\mu$ on the domain $\Omega$ to concentration with regard to the {\em empirical measure} $\mu_n$ supported on the dataset $X$ (this was essentially a criticism of \cite{pestov:00} made in \cite{shaft:06})? Here one needs the  machinery of {\em statistical learning theory} of Vapnik and Chervonenkis \cite{vapnik:98,anthbart:99,devroye:96,vidyasagar:2003}, which can guarantee such results provided the class $\mathscr F$ has low combinatorial complexity (e.g., a finite VC dimension). 
This way, one obtains $\Omega(n/d\lg n)$ lower bounds for the pivot table expected average performance \cite{VolPest09}, as well as superpolynomial in $d$ lower bounds for metric trees \cite{pestov:08b}.

{\em Approximate NN queries} \cite{IM,PC:08} seem to be in some sense free from the curse of dimensionality. In fact, the concentration of measure becomes a positive force here, and we will try to explain why, using the example of random projections in the Hamming cube (the approach of Kushilevitz, Ostrovsky and Rabani \cite{KOR:00}), as well as the Euclidean space (the Johnson--Lindenstrauss lemma \cite{JL}).

Getting back to exact search, the {\em Curse of Dimensionality Conjecture} \cite{indyk:04} calls for a general statement about lower bounds, which would apply across the entire range of all possible indexing schemes.
The conjecture is still open even for the Hamming cube $\{0,1\}^n$, and we discuss it briefly.

We conclude the article with a few remarks on the notion of intrinsic dimensionality of data, on a black-box search model of Krauthgamer and Lee \cite{KL:05}, as well as on a spatial approximation algorithm based on Delaunay graphs \cite{navarro:2002}.

\section{Concentration}

\subsection{The concentration of measure phenomenon}
Informally, the phenomenon can be stated as follows:

\begin{quote}
On a typical ``high-dimensional'' structure $\Omega$, every $1$-Lipschitz function $f\colon\Omega\to [0,1]$ has small variation.
\end{quote}

Usually, however, concentration is being dealt with using a different dispersion parameter. We proceed to precise definitions.

Let a metric space $(\Omega,\rho)$ carry a probability measure $\mu$. Such an object is called a {\em metric space with measure}.
One defines the {\em concentration function} $\alpha_{\Omega}$ of $\Omega$ by setting $\alpha_{\Omega}(0)=1/2$ and, for $\e>0$,
\[\alpha_{\Omega}(\ve)=
%\left\{
%\begin{array}{ll} \frac 12, & \mbox{$\ve=0$,} \\
1-\inf_{A\subseteq\Omega}\left\{\mu\left(A_\e\right) \colon
% A\subseteq\Omega, ~~ 
\mu(A)\geq\frac 12\right\}.
%, & \mbox{$\ve>0$.}
%\end{array}\right.
\]

The value $\alpha_{\Omega}(\ve)$ gives a uniform upper bound on the measure of the complement to the $\ve$-neighbourhood $A_{\ve}$ of every subset $A$ of measure $\geq 1/2$, cf. Fig. \ref{fig:illconc1}. 
\begin{figure}[ht]
\centering
\scalebox{0.3}[0.3]{
\includegraphics{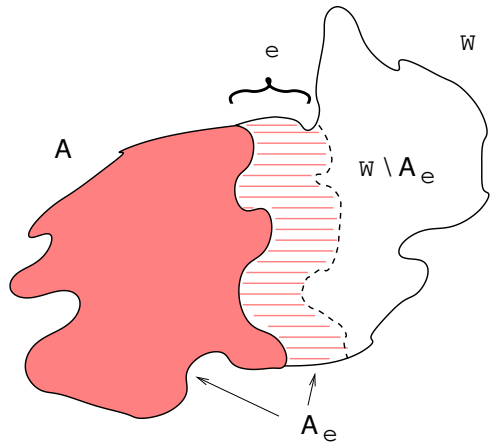}}
\caption{To the concept of the concentration function $\alpha_{\Omega}(\e)$.}
\label{fig:illconc1}
\end{figure}

On a typical high-dimensional geometric object $\Omega$ the function $\alpha(\e)$ drops off steeply near zero. For regular geometric objects such as Hamming cubes, Euclidean unit spheres and so on, one can usually derive gaussian upper bounds of the form 
\[\alpha(\e)\leq \exp(-\Theta(\e^2 d)),\]
where $d$ is the dimension parameter. 

For example, the Hamming cube $\{0,1\}^d$ with the normalized Hamming metric and uniform measure satisfies a Chernoff bound $\alpha(\e)\leq \exp(-2\e^2d)$ (obtained by combining Harper's isoperimetric inequality, see e.g. \cite{FF:81}, with the classical Chernoff bound, cf. \cite{vidyasagar:2003}, 2.2.1). See Figure \ref{fig:hamming100}.

\begin{figure}[ht]
\centering
\epsfig{file=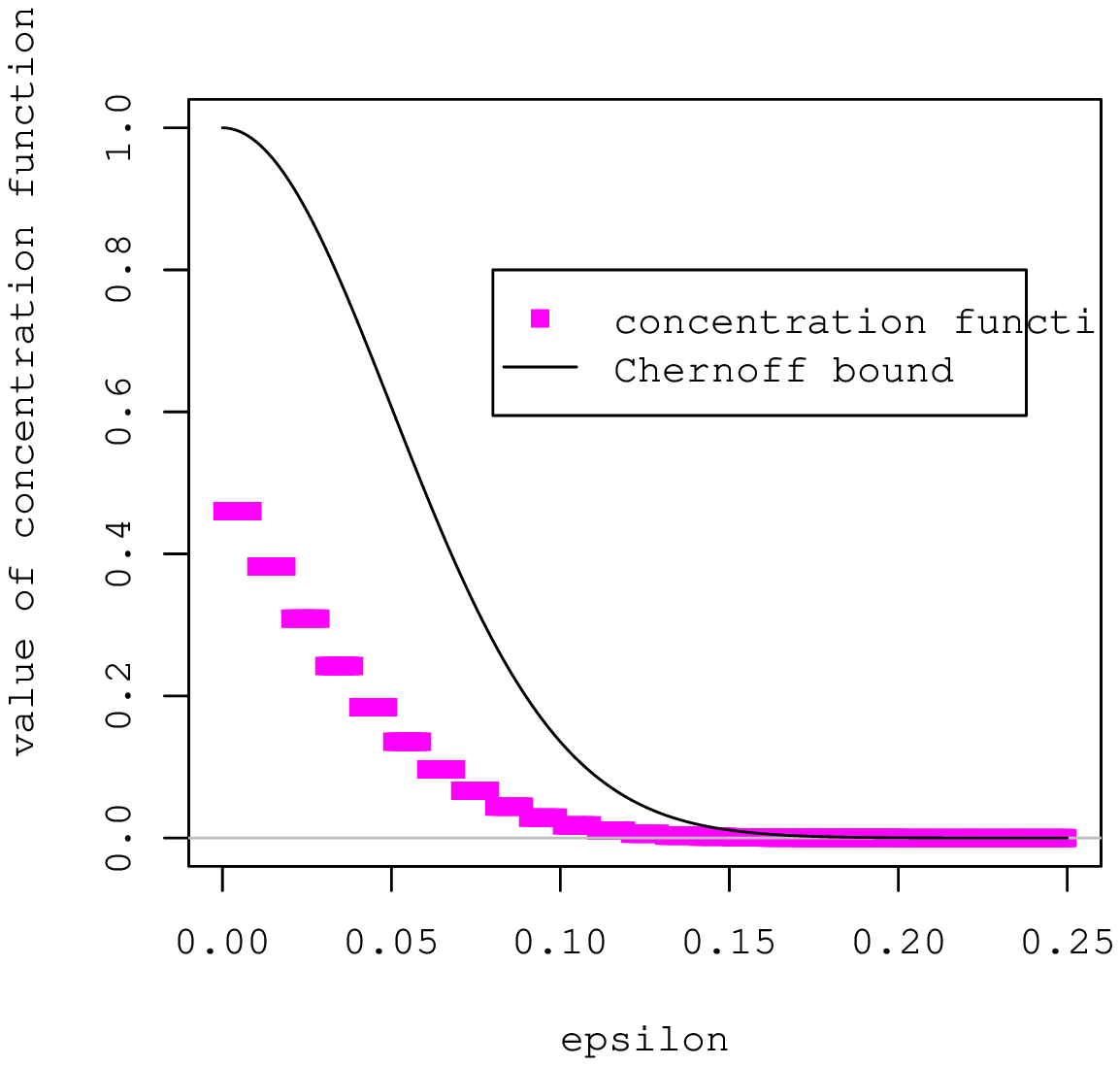, height=3.3in, width=3.3in} 
\caption{Concentration function of the Hamming cube of dimension $d=100$ {\em vs} Chernoff bound.}
\label{fig:hamming100}
\end{figure} 

It follows easily that for every real-valued $1$-Lipschitz function $f$ on $\Omega$ and for each $\ve>0$ one has
\begin{equation}
\mu\{x\in\Omega\colon \left\vert f(x) - M_f\right\vert >\ve \}\leq 2\alpha_{\Omega}(\ve),
\label{eq:mf}
\end{equation}
where $M_f$ is the median value of $f$, that is, a (generally non-unique) real number with the property that for a randomly drawn $x\in\Omega$ the probabilities of the events $[f(x)\geq M]$ and $[f(x)\leq M]$ are at least $1/2$ each. One can further derive uniform upper bounds in terms of $\alpha_{\Omega}$ on the variances of $1$-Lipschitz functions on $\Omega$ with values in a bounded interval.

The concentration phenomenon admits the following illustration. Draw $1,000$ points randomly from a high-dimensional geometric object such as the unit cube
$\I^d$ centred at the origin, choose a random orthogonal projection onto a two-dimensional subspace, and project both the cube and the chosen points on this subspace. The points will concentrate near the centre, the more so the higher the dimension $d$ is, as seen in Figure \ref{fig:cpts300} for the values of dimension $d=3,10,100$ and $1000$. The red outline is the two-dimensional projection of the cube.

\begin{figure}[ht]
\centering
\epsfig{file=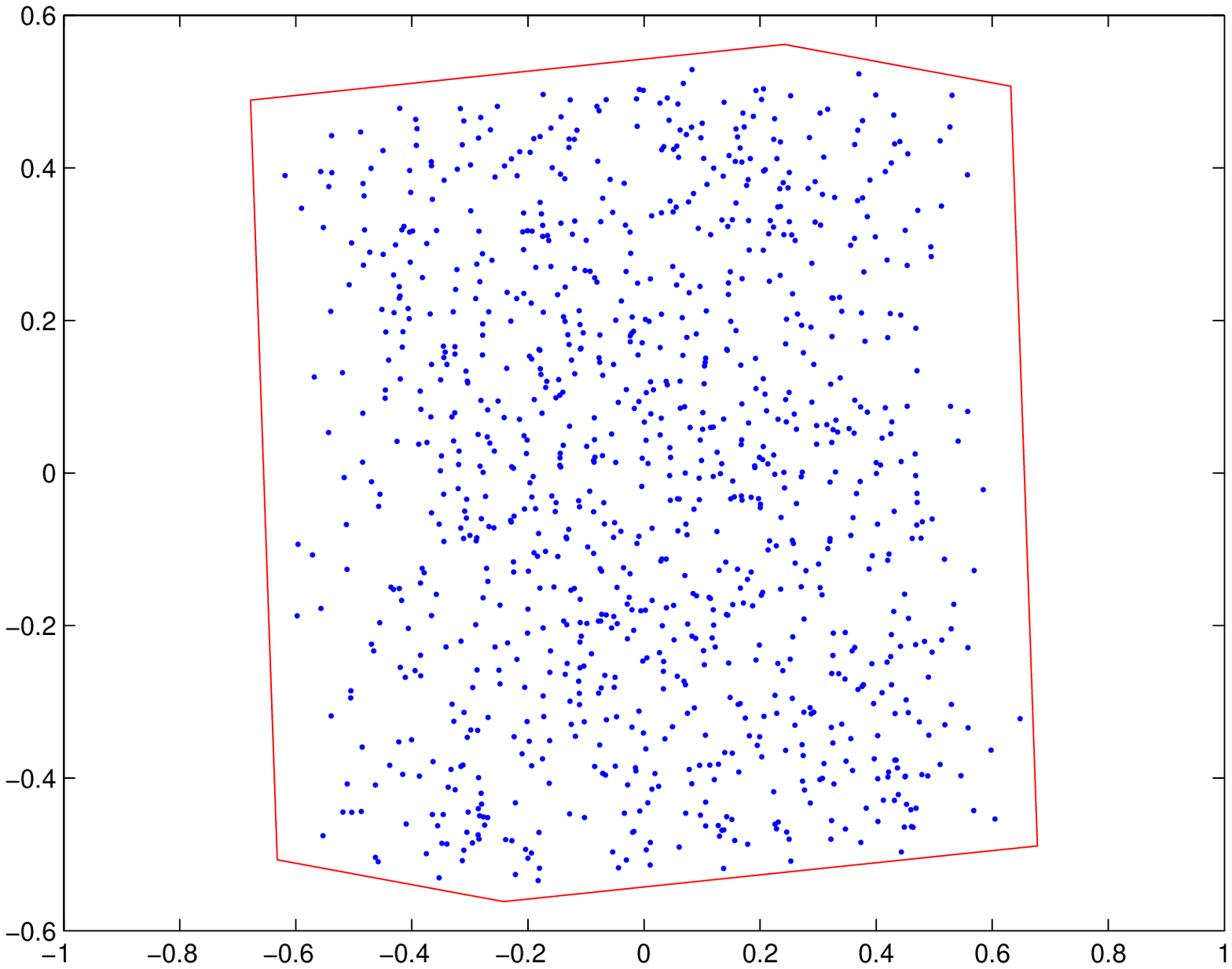, height=1.5in, width=1.5in} 
\epsfig{file=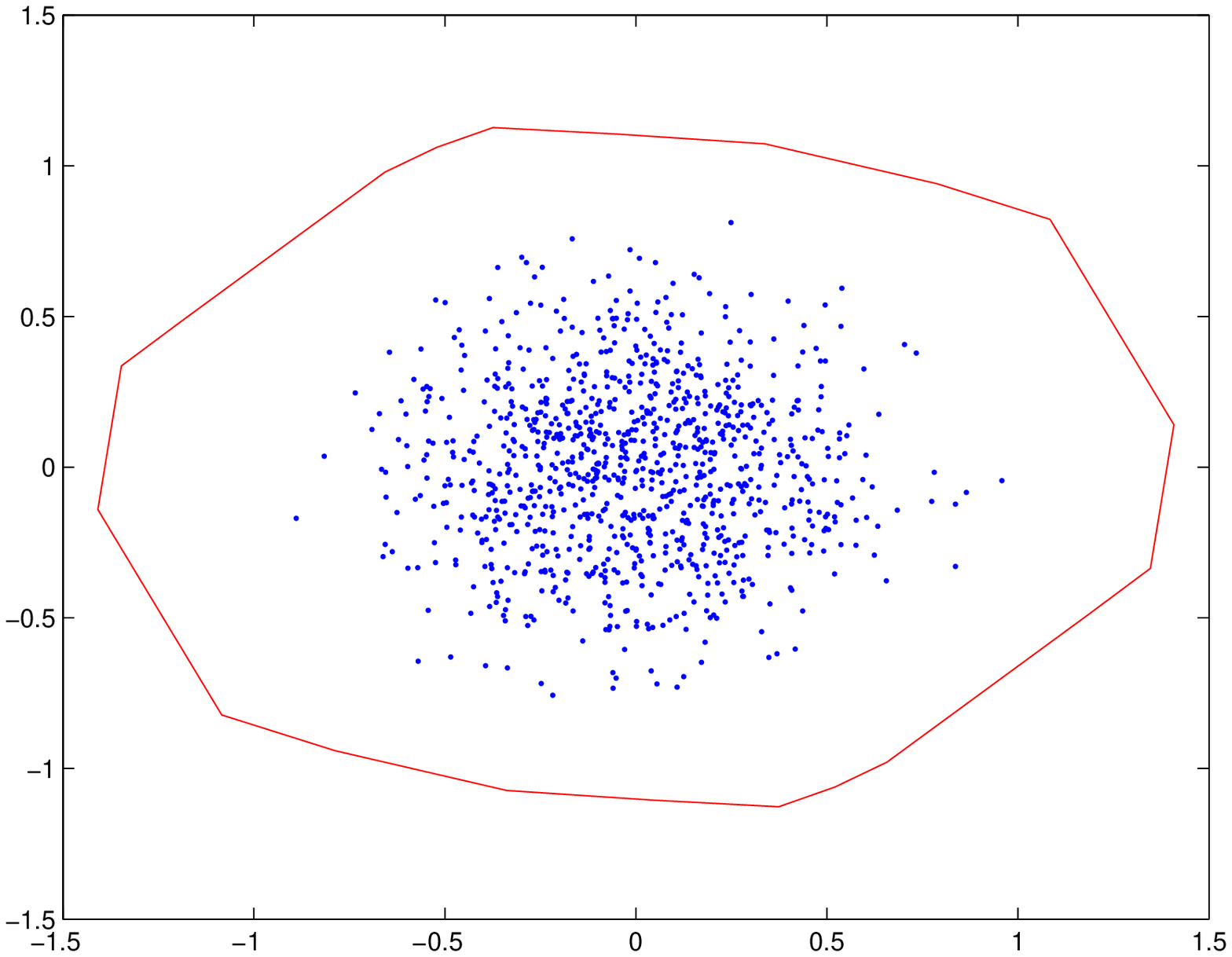, height=1.5in, width=1.5in} 
\epsfig{file=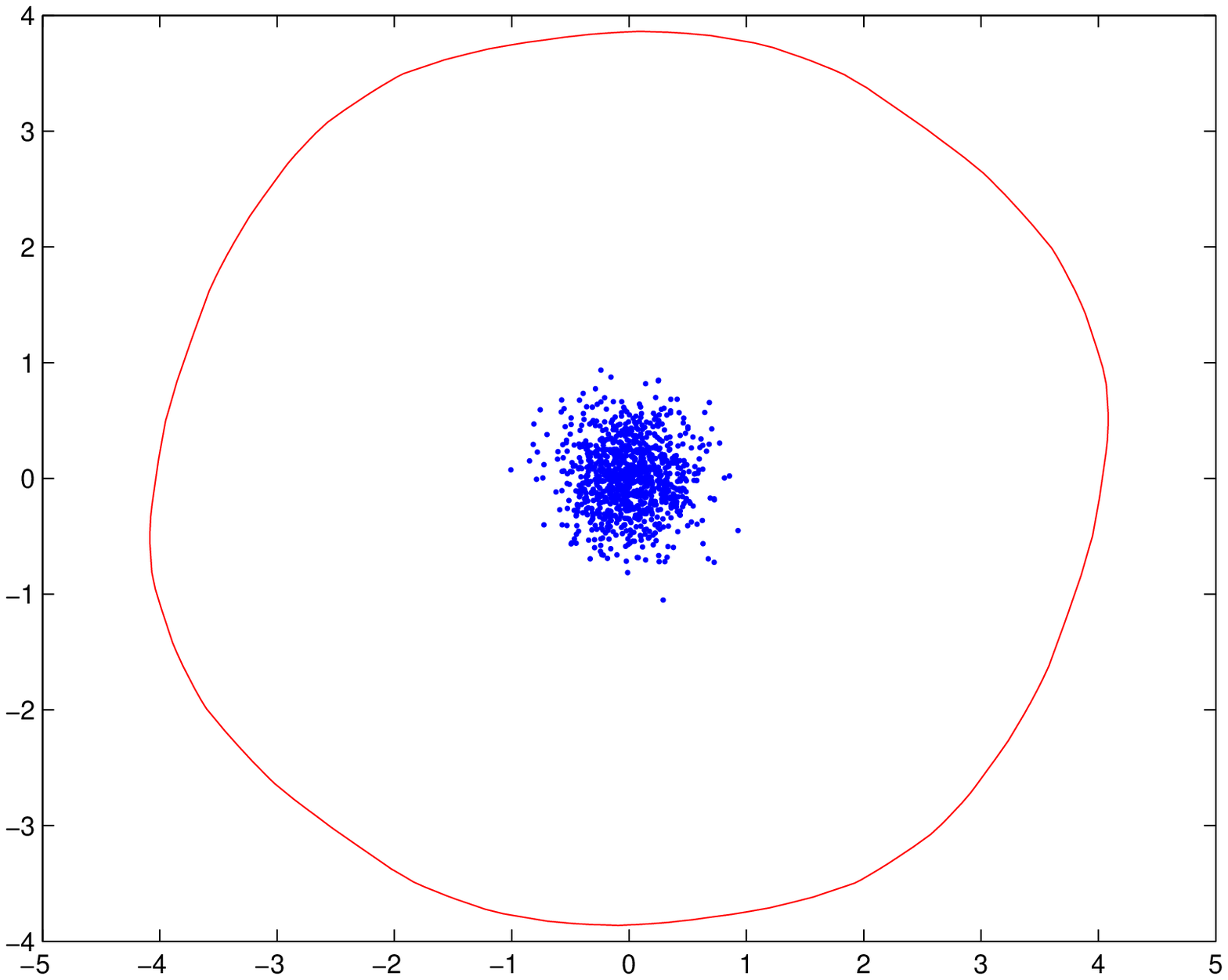, height=1.5in, width=1.5in}
\epsfig{file=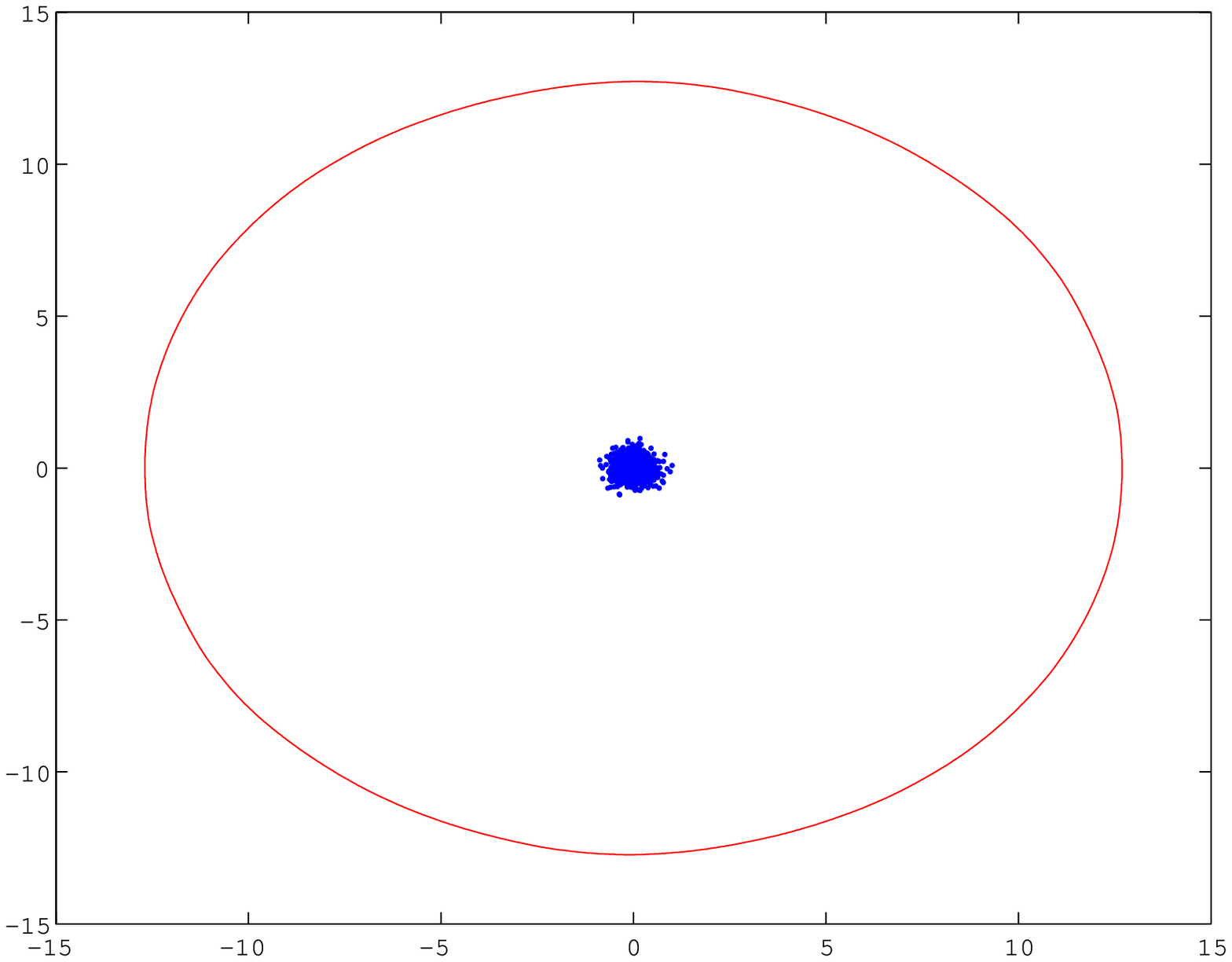, height=1.5in, width=1.5in} 
\caption{Orthogonal projection of a unit Euclidean $d$-cube and of $1,000$ random points inside the cube on a random $2$-subspace, $d=3$ (top left), $d=10$ (top right), $d=100$ (bottom left), $d=1000$ (bottom right).}
\label{fig:cpts300}
\end{figure}

Another noteworthy concequence of concentration is that the shape of the random projection of the cube is getting ever more similar to a disk as $d\to\infty$. In fact, for $d=1000$ the only visual difference one can spot between a random projection of a unit cube and that of a unit sphere, is the scale of the two projections: the diameter of a unit $d$-cube is $O(\sqrt{d})$.

This illustrates an interesting feature of geometry of high dimensions: many  high-dimensional objects look essentially the same to a low-dimensional observer. For instance, a certain precise version of this statement holds true for {\em all} convex bodies, as recently proved by Klartag \cite{klartag:07}. For this reason, for an asymptotic study of performance of indexing schemes when $d\to\infty$ the choice of a particular family of domains (Euclidean spheres, balls, cubes, Hamming cubes) does not matter that much.

Among the books treating the concentration phenomenon, \cite{MS} is the most reader-friendly, \cite{ledoux} most comprehensive, and \cite{gromov:99} contains a wealth of ideas. See also a survey \cite{M00}.

\subsection{Asymptotic assumptions on the similarity workload\label{ss:assumptions}}
Let us agree on the following four assumptions on the similarity workload:

\subsubsection{Domain as a metric space with measure} 
The metric domain $(\Omega,\rho)$ is equipped with a probability measure $\mu$, and datapoints are drawn from $\Omega$ in an i.i.d. fashion following the distribution $\mu$.

{\small
(This is the model used in \cite{CPZ:98}, which of course agrees with the traditional statistical approach to data modelling.)}

\subsubsection{Normalization of the distance} 
The distance $\rho$ on the domain is normalized so that the characteristic size of $\Omega$ is constant:
\[{\mathrm{CharSize}}(\Omega) = \E_{\mu\otimes\mu}(\rho)=O(1).\]

{\small
(Every domain can be renormalized in the above fashion unless the expected distance between the two points is infinite, which does not appear to be a realistic assumption anyway.)}

\subsubsection{Growing instrinsic dimension} $\Omega$
has ``intrinsic dimension $d$'' in the sense that the concentration function of the metric space with measure $(\Omega,\rho,\mu)$ admits a gaussian upper bound
\[\alpha_{\Omega}(\ve) = \exp(-\Omega(\ve^2d)).\]

{\small (Such an approach to intrinsic dimensionality is developed in \cite{pestov:07,pestov:08}.)}

\subsubsection{Size of a dataset} The number $n$ of datapoints grows faster than any polynomial function in $d$, but slower than any exponential function in $d$:
\begin{eqnarray}
n=d^{\omega(1)},~~
d=\omega(\log n).  
\end{eqnarray}
{\small
(This is a standard assumption in the asymptotic analysis of indexing schemes for similarity search, cf. \cite{indyk:04}. An example of such a rate of growth is $n=2^{\sqrt d}$.)}

Note that randomly drawing a single dataset $X\subseteq\Omega$ with $n$ points amounts to randomly drawing a single point in the $n$-th power of the domain, $\Omega^n$, equipped with the product probability measure $\mu^{\otimes n}$.
In order to perform asymptotic analysis of indexing scheme performance, we will in fact be choosing an {\em infinite sequence} of datasets $X_d\subseteq \Omega_d$, $d=1,2,\ldots$. This is equivalent to drawing a single point $\bar x$ ({\em sample path}) in the infinite product
\[\Omega_1^{n_1}\times\Omega_2^{d_2}\times\ldots\times\Omega_d^{n_d}\times\ldots,\]
with regard to the corresponding infinite product of probability measures:
\[\bar X\sim \mu_1^{\otimes n_1}\otimes\mu_2^{\otimes d_2}\otimes\ldots\otimes\mu_d^{\otimes n_d}\otimes\ldots.\]
When talking about {\em confidence,} we will mean the product probability in the above infinite product space. Specifically, a statement $Q(d,\bar x)$ parametrized by the dimension $d$ and taking as a variable the sample path $\bar x$ {\em occurs with} ({\em asymptotically}) {\em high confidence} if for every $\delta>0$ there is $D$ so that 
\[P[Q(d,\bar X)\mbox{ is true }]>1-\delta\]
whenever $d\geq D$. 

At the same time, in order to keep the notation simple, we will suppress the dimension index $d$ and talk just of a single domain $\Omega$ and a dataset $X\subseteq\Omega$.

\subsection{Empty space paradox\label{ss:empty}} 
Denote $\ve_{NN}$ the nearest neighbour distance function on $\Omega$, given by $\ve_{NN}(\omega) = \rho(\omega,X)$.

\begin{theorem} 
\label{th:empty}
Under our standing assumptions on the workload, for every $\e>0$ one has with asymptotically high confidence that for all points $\omega\in\Omega$ except for a set of measure $\exp(-\Omega(\e^2 d))$
\[\left\vert\ve_{NN}(\omega)-{\mathrm{CharSize}}(\Omega)\right\vert<\e.\]
\end{theorem}

The result applies to the Hamming cube, the Euclidean cube, the Euclidean space with gaussian measure, the Euclidean ball, etc.
% For a proof, see \cite{pestov:08b}, Lemma 1.

As a byproduct of the technique, one obtains:

\begin{proposition}
\label{p:simplex}
Under the same assumptions, for every $\e>0$ the pairwise distances between datapoints of $X$ are all in the range ${\mathrm{CharSize}}\,(\Omega)\pm \ve$ with asymptotically high confidence.
\end{proposition}

For constant $n=\abs X$ and the case of a Euclidean domain the result was established in \cite{HMN}.

% > postscript(file="hamming100.eps",width=5,height=5,horizontal=FALSE,onefile=FALSE)
% > x <- seq(0, 0.25, length= k)
% > y <- pbinom((N/2)+floor(x*N),size=N, prob=0.5,lower.tail=FALSE)
% > plot(x,y, type="p",pch=15,col="6",xlim=c(0,0.25),ylim=c(0,1),xlab="epsilon",ylab="value of concentration function", main="")
% > abline(h=0, col="gray")
% > legend(0.08, 0.8, c("concentration function", "Chernoff bound"),pch=c(15,-1),lty=c(0,1),col=c(6,1))
% > z <- exp(-2*x^2*N)
% > lines(x,z)
% > dev.off()

Our proofs can be found in Appendix A.
% We include the proofs of the two above results in Appendix A.

\section{VC theory}
\subsection{VC dimension\label{ss:vc}}
Let $\mathscr C$ denote a collection of subsets of the domain $\Omega$. The VC dimension is an important measure of combinatorial complexity of $\mathscr C$.
A finite set $A\subseteq\Omega$ is {\em shattered} by $\mathscr C$ if every subset $B\subseteq A$ can be ``carved out'' of $A$ with the help of a suitable element $C$ of $\mathscr C$:
\[B = A\cap C.\]
\begin{figure}[ht]
\begin{center}
\scalebox{0.25}[0.25]{\includegraphics{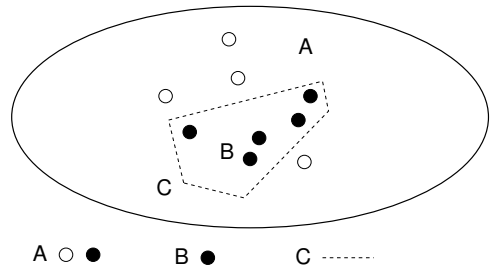}} 
\caption{To the concept of a set $A$ shattered by a class $\mathscr C$.}
\end{center}
\end{figure}  

The {\em VC dimension} of $\mathscr C$, denoted $\VC({\mathscr C})$, is the supremum of cardinalities of all finite subsets of the domain which are shattered by $\mathscr C$.
Here are some classical examples.

\begin{center}
\begin{tabular}{l|l}
\hline
Family of sets & VC dimension \\
\hline \hline
Intervals in $\R$ & $2$ \\
Half-spaces in $\R^d$ & $d+1$ \\
Euclidean balls in $\R^d$ & $d+1$ \\
Parallelepipeds in $\R^d$ & $2d+2$ \\
Convex polygons in $\R^d$ & $\infty$ \\
Any family with $n$ sets & $\leq\lg_2n$ \\
Hamming balls in $\{0,1\}^d$ & $\leq d + \lg_2d$ \\
%Finite unions of intervals in $\R$ & $\infty$ \\
\hline
\multicolumn{2}{c}{
If $\VC({\mathscr C})=c$, $\VC({\mathscr D})=d$:} \\
\hline
$\{\Omega\setminus C\colon C\in {\mathscr C}\}$ & c \\
$\mathscr C\cup\mathscr D$ & $\leq c+d+1$ \\ 
$k$-fold intersections & $\leq 2k\lg(ek)c$\\ 
~~~of members of $\mathscr C$ &
\\
\hline
\end{tabular}
\end{center}

Proofs can be found e.g. in \cite{vidyasagar:2003}, Ch. 4.

Estimating the VC dimension of a particular family of sets is often a non-trivial task. For example, the value of this parameter does not seem to be known for the collection of all cubes in $\R^d$ with sides parallel to the coordinate hyperplanes. More generally, it is tempting to conjecture that the VC dimension of the family of all balls (either open or closed) in a Banach space of finite dimension $d$ equals $d+1$, but the author is unaware of any results beyond the Euclidean case.

\subsection{Uniform convergence of empirical measures}
Recall that the {\em Borel sigma-algebra} of subsets of a metric space $\Omega$ is the smallest family closed under countable intersections and complements and containing all open balls. Elements of the Borel sigma-algebra are called simply {\em Borel subsets}. We will restrict our attention to those families $\mathscr C$ whose elements are Borel subsets of $\Omega$. This assumption guarantees that the value $\mu(C)$ is well-defined for every probability measure $\mu$ on $\Omega$. 

The {\em empirical measure} of $C\in {\mathscr C}$ with regard to a finite sample $X=\{x_1,\ldots,x_n\}$ is just the normalized counting measure 
\[\mu_n(C) = \frac 1n\left\vert \{i\colon x_i\in C\}\right\vert.\] 
The VC dimension of $\mathscr C$ is finite if and only if, with high confidence, the empirical measures of every $C\in {\mathscr C}$ converge uniformly to the true value $\mu(C)$ as the sample size $n$ goes to infinity, no matter what the underlying measure $\mu$ is.

Here is a more exact formulation. A class $\mathscr C$ has the property of {\em uniform convergence of empirical measures}, or is a {\em uniform Glivenko--Cantelli class,} if there is a function $s(\delta,\e)$ ({\em sample complexity} of the class) so that, given a desired precision value $\e>0$ and a risk level $\delta>0$, whenever $n\geq s(\delta,\e)$, one has
\[\sup_{\mu\in P(\Omega)}P\left\{\sup_{C\in{\mathscr C}}\left\vert \mu(C) - \mu_n(C) \right\vert\geq \e\right\}<\delta.
\]
Here $P(\Omega)$ denotes the family of all probability measures on $\Omega$. We quote the following as stated in \cite{vidyasagar:2003}, Theorem 7.8.

\begin{theorem}[Uniform Glivenko--Cantelli theorem]
\label{th:ugc}
A concept class $\mathscr C$ is uniform Glivenko--Cantelli if and only if $d=\VC({\mathscr C})<\infty$, in which case
\[s(\delta,\e) \leq \max\left\{\frac{8d}{\e}\lg\frac{8e}{\e},\frac 4{\e}\lg\frac{2}{\delta} \right\}.
\]
\end{theorem}
One of the components of the proof is the concentration of measure in the Hamming cube $\{0,1\}^n$.

Let us remark that similar results can be stated and proved for {\em function classes,} that is, collections $\mathscr F$ of functions from the domain $\Omega$ to the interval $[0,1]$. The role of VC dimension is taken over by other combinatorial parameters, such as the {\em fat shattering dimension}. We will not enter into details.

Among a great selection of books treating VC theory, let us mention encyclopaedic sources \cite{vidyasagar:2003} and \cite{devroye:96}, a classical monograph  \cite{vapnik:98}, and a lighter, but very well-written \cite{anthbart:99}.

\section{The curse of dimensionality}
\subsection{Pivot tables}
\subsubsection{Reduction and access overhead}
Let $(\Omega,\rho,X)$ be a similarity workload, $\Upsilon$ a metric space, and $f\colon\Omega\to\Upsilon$ a $1$-Lipschitz function. If queries in $\Upsilon$ are easier to process than in $\Omega$, then it makes sense, given a range query $(q,\e)$ in $(\Omega,X)$, to run a $(f(q),\e)$ range query in $(\Upsilon,f(X))$, retrieving all datapoints $x$ with $f(x)$ within the distance of $\e$ of $f(q)$, and then check them against the condition $\rho_{\Omega}(q,x)<\ve$. The $1$-Lipschitz property of $f$ guarantees that no true hits will be missed.

In this way, the function $f$ can be viewed as a {\em projective reduction} of the exact similarity search problem to the new workload $(\Upsilon,f(X))$. This viewpoint is developed in some detail e.g. in \cite{PeSt06}.  The {\em access overhead} of the reduction $f$ is defined as
\[\mbox{acc}_f(q) = \abs{X\cap f^{-1}(B_{\ve}(f(q)))} - \abs{X\cap B_{\ve}(q)}.\]
This simple and well-known idea on its own can be surprisingly efficient, cf. \cite{StPe07}.

\subsubsection{Pivot-based reduction to $\ell^\infty(k)$}
\label{s:pivot}
Every finite collection $f_1,f_2,\ldots,f_k$ of 1-Lipschitz functions on  $(\Omega,\rho)$ defines a $1$-Lipschitz mapping $f=\Delta_{i=1}^k f_i$ from $\Omega$ to $\ell^\infty(k)$ via the formula
\[f(x) = (f_1(x),f_2(x),\ldots,f_k(x)).\]
Here $\ell^\infty(k)$ is the vector space $\R^k$ equipped with the norm $\norm x_{\infty} = \max_{i=1}^k\abs{x_i}$. 
If the $f_i$ are distance functions from pivot points $p_i\in\Omega$,
the resulting mapping $f$ is of the form
\begin{equation}
\label{eq:map}
f(x) =(d(x,p_1),\ldots,d(x,p_k))\in\ell^\infty(k).\end{equation}
In \cite{vidal:86}, it was suggested to use a reduction of this form in case where the distance computations in $\Omega$ are so expensive that even a simple sequential scan of the image $f(X)$ in $\ell^\infty(m)$ is computationally cheaper. This idea was analyzed for more general similarity measures than metrics in \cite{farago:93}. By combining it with other access methods on the space $\ell^\infty(m)$, further new indexing methods have been developed, see e.g. \cite{chavez:99}. 

A $m$-NN similarity query is processed in $(\Omega,d,X)$ in time
\[k+\ell+(\mbox{acc}_f(q)+m).\]
Here the first term stands for the calculation of $k$ distances from a query point $q$ to the pivots and $\ell$ is the processing time of a rectangular query in $\ell^\infty(k)$, while the latter expression lists the number of distance computations in $\Omega$ needed to separate false hits from $k$ true positives. %There is only one term above that can be minimized through the choice of the pivots. 
A classical paper on optimizing the pivot selection is \cite{BNC:03}.

\subsubsection{Lower query time bounds for pivot tables}
Our next result (a slightly corrected version of the main theorem in \cite{{VolPest09}}) is valid not only for the Hamming cube which is a testbed for asymptotic analysis of performance of indexing schemes, but also for the Euclidean space $\R^d$ with the gaussian measure, the cube $[0,1]^d$, and so forth.

\begin{theorem}
\label{th:pivotcurse}
In addition to the assumptions of Subs. \ref{ss:assumptions}, suppose also that the VC dimension of the family of all balls in $\Omega$ is $O(d)$. Any pivot table with $k=o(n/d\log n)$ pivots will return an expected average number of $\Omega(n)$ datapoints. Consequently, the average total complexity of the performance of any pivot table for the resulting workload is $\Omega\left({n}/{d\log n}\right)$.
\end{theorem}

\begin{proof}
Assume the number of pivots $k$ is $o(n/d\log n)$.
Let $\ve_M$ denote the median value of the function $\ve_{NN}$, so that for at least half query points $q$ the distance to the NN in $X$ is $\geq\ve_M$. For each pivot $p_i$, $i=1,2,\ldots,k$, denote $\rho_i^M$ the median value of the distance function $\rho(p_i,-)$. 
Because of concentration, the mesure of the spherical shell 
\[S_i = \{q %\in\Omega
\colon \rho_i^M-\ve_M/2<\rho(p_i,q)<\rho_i^M+\ve_M/2\}\]  is $1-\exp(-\Omega(\ve_M^2 d))$, and the complement to the intersection, $S=\cap_iS_i$, of all $k$ shells has measure
\[
o(n/d)\exp(-\Omega(\ve_M^2 d))= \exp(-\Omega(\ve_M^2 d)),\]
since $n$ is subexponential in $d$.  
\begin{figure}[ht]
\begin{center}
\scalebox{0.25}[0.25]{\includegraphics{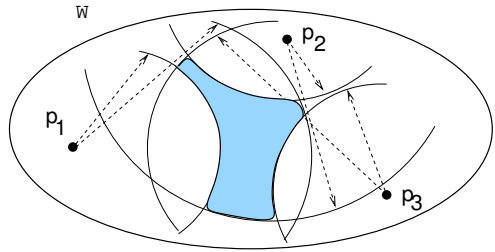}} 
\caption{An intersection of spherical shells.\label{fig:pivots}}
\end{center}
\end{figure} 
Thus, among all $k$-fold intersections of spherical shells (Figure \ref{fig:pivots}), we have found a giant one, whose $\mu$-measure is nearly one.

To assure that this intersection contains an accordingly high proportion of datapoints, consult the table in Subs. \ref{ss:vc} to deduce that the family of all $k$-fold intersections of spherical shells in $\Omega$ has VC dimension not exceeding $2k\lg(ek)O(d)=o(n)$. By Theorem \ref{th:ugc}, the empirical measure $\mu_n(S)$ approaches $\mu(S)$ and therefore $1$ with high confidence as $d\to\infty$. 

The measure of the set $Q$ of query points $q\in S$ whose distance to the nearest neighbour in $X$ is greater than or equal to $\ve_M$ is at least $1/2 - \exp(-O(d)\e^2)$. For every non-empty range query $(q,\e)$ where $q\in Q$, all datapoints belonging to $S$, that is, most datapoints of $X$, have to be returned. This gives an expected average total complexity $\Omega(n)$ under our assumption on the number of pivots.
\end{proof}

Notice that we allow the pivots $p_i$ to be arbitrary points of the domain $\Omega$. If we require that pivots be chosen from the dataset $X$, then the set $S$ in the above proof will with high confidence contain $n-k$ datapoints by Theorem \ref{th:empty} and Proposition \ref{p:simplex}, and we obtain (without using VC theory):

\begin{corollary} Under the assumptions of Subs. \ref{ss:assumptions}, if all pivots $p_i$ belong to the dataset $X$, then the expected total complexity of the performance of the resulting pivot table is $n(1-o(n))$.
\end{corollary}
   
\subsubsection{A remark on results of \cite{farago:93}}
The above lower bounds agree with an exponential in $d$ upper bound of $k+c^d$ derived in the influential paper \cite{farago:93}, Theorem 3 within a similar model, with no restriction on a number $n$ of datapoints, and with $d$ a dimension parameter defined by a certain measure distribution density condition verified, e.g. by the Hamming cube $\{0,1\}^d$ or the Euclidean sphere $\s^d$. Here $c$ is a constant depending on $\Omega$, the smallest distortion parameter of a $1$-Lipschitz embedding $f\colon\Omega\to\ell^{\infty}(k)$:
\[
\forall x,y\in X,
\]
\[
\norm{f(x)-f(y)}_{\infty}\leq\rho(x,y)\leq c\norm{f(x)-f(y)}_{\infty}.
\]
However, the usefulness of the result is limited because of an imprecise claim ({\em loc. cit.}, Example 1) that for a bounded subset $X$ of $\ell^2(d)$ there always exists a 1-Lipschitz function $f\colon X\to\ell^\infty(d+2)$ having distortion $c\leq 2$. In fact, an optimal constant here is on the order $\Omega(\sqrt d)$ (see \ref{s:distortion}). As a result, the query performance estimate for the Euclidean domains made in Remark after the main Theorem 3, {\em loc.cit.}, becomes superexponential in $d$ and thus meaningless.

This misconception has led to some further confusion, cf. remarks made in \cite{BNC:03} (p. 2358, end of first paragraph on the r.h.s., and at the beginning of Section 5).

\subsection{Hierarchical metric tree schemes}
\subsubsection{Metric trees}
For a finite rooted tree $T$ we denote $L(T)$ the set of leaves of $T$ and $I(T)$ the set of inner nodes. The symbol $\ast$ will denote the root node of $T$.

Let $\mathscr F$ be a class of $1$-Lipschitz functions on $\Omega$ (possibly partially defined). 
A {\em metric tree} ({\em of type $\mathscr F$}) for a workload $(\Omega,\rho,X)$ is a hierarchical indexing structure consisting of\\[1mm]
$\bullet$  a finite binary rooted tree $\mathcal T$,\\
$\bullet$ an assignment of a function $f_t\in{\mathscr F}$ (a {\em pruning,} or {\em decision function}) to every inner node $t\in I(T)$, and\\
$\bullet$  a collection of subsets $B_t\subseteq\Omega$, $t\in L(T)$ ({\em bins}), covering the dataset: $X\subseteq\cup_{t\in L(T)}B_t$.

Since we assume that the tree $T$ is binary, it can be identified with a sub-tree of the prefix tree, that is, a subset of binary strings $\e_1\e_2\ldots\e_k$, $0\leq k\leq n$, where $\e_i=\pm 1$ for all $i$. 
% The indexing scheme is constructed so as to assure the following condition. Let $s=\e_1\e_2\ldots\e_m$ be a leaf node, and let $x$ belong to the bin $B_s$ labelled with $s$. Then for each $0\leq k\leq m$, the value of $f_{\e_1\e_2\ldots\e_i}$ at $x$ equals $\e_{i+1}$. 

At each inner node $t=\e_1\e_2\ldots\e_l$ the value of the pruning function $f_t$ at the query center $ q $ is evaluated. The condition $f_t( q )>\e$ gurantees that the child node $t(-1)=\e_1\e_2\ldots\e_l(-1)$ need not be visited, because all elements $x$ of the bins indexed with the descendants of $t(-1)$ are at a distance $>\e$ from $ q $. Indeed, assuming $x\in B_\e( q )$, one has
\[\abs{f_t(x)-f_1( q )}\leq d(x, q )\leq\e.
\]
Similarly, if $f_t( q )<-\e$, then the node $t1=\e_1\e_2\ldots\e_k1$ can be pruned, because no bin labelled with descendants of $t1$ can possibly contain a point within the range $\e$ from $ q $. 

However, if $f_t( q )\in [-\e,\e]$, then no pruning is possible and both children nodes of $t$ have to be visited. The search branches out. In the presence of concentration, the amount of branching is considerable, and results in dimensionality curse. 

The M-tree \cite{CPZ:97} is by now a classical example of a metric tree. However, metric tree-type indexing schemes are very numerous, cf. Sections 2.1-2.4 in \cite{ZADB:06} or Section 4.5 in \cite{samet}.

\subsubsection{Lower bounds for metric trees}

For a function $f$ and a real number $t$, denote $\mathbf{1}_{f\leq 1}=\{x\in\mathrm{dom}(f)\colon f(x)\leq t\}$.

\begin{theorem}
\label{th:mtree}
In addition to the assumptions of Subs. \ref{ss:assumptions}, let $\mathscr F$ be a class of $1$-Lipschitz functions on the domain $\Omega$ such that the VC dimension of the family of sets $\mathbf{1}_{f\leq t}$, $f\in {\mathscr F}$, $t\in\R$ is ${\mathrm{poly}}(d)$. Then the expected average performance of every metric tree indexing structure of type $\mathscr F$ is superpolynomial in $d$.
\end{theorem}

That the above combinatorial assumption on the class $\mathscr F$ is sensible, follows from a theorem of Goldberg and Jerrum \cite{GJ}. Consider a parametrized class
\[{\mathscr F}=\{x\mapsto f(\theta,x)\colon\theta\in \R^s\}\]
for some $\{0,1\}$-valued function $f$. Suppose that, for each input $x\in\R^s$, there is an algorithm that computes $f(\theta,x)$, and this computation takes no more than $t$ operations of the following types:
\\
$\bullet$ the arithmetic operations $+,-,\times$ and $/$ on real numbers,
\\
$\bullet$
jumps conditioned on $>$, $\geq$, $<$, $\leq$, $=$, and $\neq$ comparisons of real numbers, and
\\
$\bullet$ output $0$ or $1$.

Then $\VC({\mathscr F})\leq 4s(t+2)$. 

Essentially, the above result states that a class of binary functions that can be computed in polynomial time taking a parameter value of polynomial length will have a polynomial VC dimension.

\hskip .3cm
{\sc On the proof of Theorem \ref{th:mtree}.} (For details, see \cite{pestov:08b}.) Suppose the conclusion is false, and fix a particular ${\mathrm{poly}(d)}$ rate, $f(d)$, bounding from above the performance of a metric tree on any sample path. 
As the total content of bins $B_t$ indexed with strings $t$ of length exceeding the rate $f(d)$ has to be asymptotically negligible, we can assume without loss in generality that the indexing tree has depth $f(d)$. 

Without loss in generality every bin can be replaced with an intersection of a family of sets of the form $\mathbf{1}_{f\leq t}$, $f\in {\mathscr F}$ and their complements. This provides a ${\mathrm{poly}(d)}$ upper bound on the VC dimension on the family of all possible bins. 

With high confidence, a bin of a large measure will contain many data points, contradicting the ${\mathrm{poly}(d)}$ performance bound. This leads to conclude that measures of bins cannot be too skewed. Now concentration of measure is used to prove that at least $\mbox{poly}(d)$ bins $B_t$ have size so large that the $\e_M$-neighbourhood of $B_t$ has almost full measure. One deduces further that query centres $q$ whose $\e_{NN}$-neighbourhood meets at least $d^{\omega(1)}$ bins have measure $\geq 1/2-o(1)$. Processing a nearest neighbour query with such a centre $q$ requires accessing all of these bins, let even to verify that some of them are empty. This leads to a contradiction with the assumed uniform performance bound on the algorithm. \qed

\subsection{The curse of dimensionality conjecture}

\subsubsection{The problem}
Of course the above are just particular results only applicable to specific indexing schemes. If one wants to validate the curse of dimensionality once and for all, here is an interesting open problem. 

\begin{conjecture}[cf. \cite{indyk:04}]
Let $X$ be a dataset with $n$ points in the Hamming cube $\{0,1\}^d$.
Suppose $d=n^{o(1)}$ and $d=\omega(\log n)$.  
Then any data structure for exact nearest neighbour search in $X$,
with $d^{O(1)}$ query time, must use $n^{\omega(1)}$ space.
\end{conjecture}

The data structure and algorithm are understood in the sense of the {\em cell probe model} of computation (cf. \cite{miltersen,borodin:99}). 
% It is considered a natural choice when one is only interested in lower bounds on the performance of general indexing schemes.

\subsubsection{Cell probe model} 
In the context of similarity search, the model can be described as follows.
An abstract indexing structure for a domain $\Omega$ consists of
\\[2mm]
$\bullet$ 
a collection of cells $C_i$, indexed with a set $I$,  
\\[1mm]
$\bullet$
a dictionary $T=W^\ast$ over an alphabet $W=\{0,1\}^b$,
viewed as a rooted prefix tree,
\\[1mm]
$\bullet$
a computable mapping $t\mapsto i(t)$  from $T$ to $I$ ({\em cell selector}), and
\\[1mm]
$\bullet$ a computable function $f=f_t(\sigma;q)$ (either partially or fully) defined on $T\times \{0,1\}^b\times \Omega$ and taking values in $W$.
\par
% The mapping $t$ is not expected to be one-to-one. In fact an acceptable size of the prefix tree $T$ is $\exp{\mbox{poly}(d)}$, while for $I$ it is only $\mbox{poly}(n)$.

For a $t\in T$, one can think of each $f_t$ as a function defined on a subset of $\Omega$ and taking a $b$-bit string $\sigma$ as a parameter, except if $t=\ast$ is the root. A value $f_t(\sigma;q)$ is a child $s$ of the node $t$. 

For every $i$, the cell $C_{i}$ can hold a $b$-bit string. Sometimes $b$ is regarded as constant, but often it is assumed that $b=\Theta(\lg n)$, so that a cell corresponding to a leaf node can store a pointer to a datapoint $x\in X$. 
Occasionally the nearest neighbour problem is replaced with a weaker {\em decision version} (known as {\em near neighbour problem}), whereby a range parameter $\ve_0>0$ is fixed and the algorithm is expected to tell whether there is an $x\in X$ at a distance $<\ve_0$ from the query point. In such a case, a leaf node cell $C_{i}$ will hold a single bit (a ``yes'' or ``no'' answer). 

Building the data structure at the preprocessing stage, given a dataset $X$, consists in storing in every node cell a $b$-bit string.

A memory image of the indexing structure $C_i,i\in I$ is created when the algorithm is initialized. 
Given a query point $q\in\Omega$, the prefix tree $T=W^\ast$ is traversed down to the leaf level beginning with the root. 
At the inner node $t$, the content $\sigma$ of the cell $C_{i(t)}$ is read and passed on to the function $f_t$ as a parameter. The computed value $f_t(\sigma;q)=s\in W=\{0,1\}^b$ indicates a child of $t$ to follow at the next step. When a leaf $l$ is reached, the algorithm halts and returns the contents of $C_{i(l)}$. The query time is the length of the branch traversed, or equivalently the number of cell probes during the execution of the algorithm. The space requirement of the model is the total number of cells, $\abs I$.

The cell probe model is very liberal, as the cost of computing the values of $f$ is disregarded. For this reason, any lower bound obtained under the cell probe  will likely hold under any other model of computation. 
% Since the nearest neighbour problem is stronger than the near neighbour problem, any lower bound for the decision version will serve as a lower bound for the curse of dimensionality conjecture.

% As an example, the pivot table will be encoded within the cell probe model in a rather cumbersome way, because coordinates of each pivot $p$ in the Hamming cube will occupy $\lceil d/b\rceil$ cells, and the distance function $\rho(p,-)$ will need to be computed recursively in $O(d/b)$ steps, with values of partial distances on substrings stored and then passed on further as parameters. For the same reason, the lower bound on the query time of the pivot algorithm proved in Theorem \ref{th:pivotcurse} becomes $\Omega(n/\lg n)$.

\subsubsection{Current state of the problem}
The best lower bound currently known is $O(d/\log\frac{sd}{n})$, where $s$ is the number of cells used by the data structure \cite{PT}. In particular, this implies the earlier bound $\Omega(d/\log n)$ for polynomial space data structures \cite{barkol:00}, as well as the bound $\Omega(d/\log d)$ for near linear space (namely $n\log^{O(1)}n$). 

\section{Approximate NN search and dimensionality reduction}

Approximate nearest neighbour search \cite{PC:08} is often said to be free from the curse of dimensionality, and the reason is that the (dimensionality) reduction maps $f$ used in indexing are no longer $1$-Lipschitz. Rather, they are what may be called ``probably approximately $1$-Lipschitz'', and sometimes only on a certain distance scale. Such maps no longer exhibit a strong concentration around their means. The price to pay is that we may lose some relevant datapoints, as some distances are typically getting distorted, and so such maps cannot be used for exact NN search.
%Curiously, the construction of reduction maps is often based on the concentration phenomemon and/or the VC theory.

\subsection{Random projections in the Hamming cube}

Think of the Hamming cube $\{0,1\}^d$ as the set of all binary functions in the space $\ell^1(d)=L^1([d])$, where $[d]=\{1,2,\ldots,d\}$ supports a uniform measure. In other words, we normalize the Hamming distance $d(x,y) = \sharp\{i\colon x_i\neq y_i\}$ by multiplying it by $1/d$. Of course such a normalization has no effect on similarity search.
If the dataset $X\subseteq\{0,1\}^d$ contains $n$ points, then the VC dimension of $X$, viewed as a concept class on $\{1,2,\ldots,d\}$, does not exceed $\lg_2n$. 
According to the uniform Glivenko--Cantelli Theorem \ref{th:ugc}, if $O(\ve^{-2}\lg_2n)$ coordinates of the Hamming cube are chosen at random, then with high confidence the restriction mapping from $X$ to the Hamming cube $\{0,1\}^{O(\ve^{-2}\lg_2n)}$ (under its own normalized Hamming distance)
preserves the pairwise distances to within $\pm\ve$.
Cf. Figure \ref{fig:distortion_h500}.

\begin{figure}[ht]
\centering
\epsfig{file=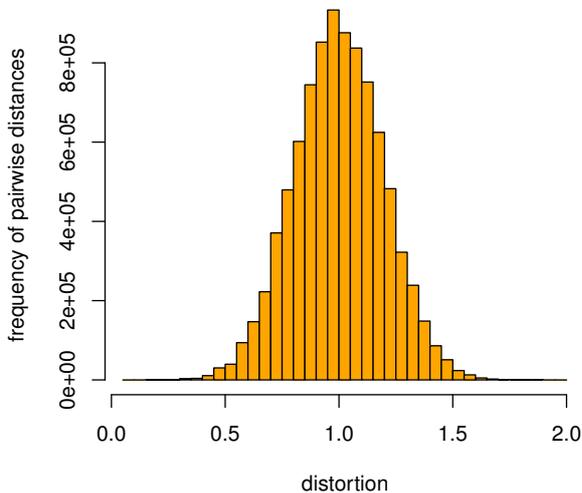, height=3.3in, width=3.3in} 
\caption{Histogram of distortions of all pairwise distances in a random dataset of $n=3,000$ points in the $d=500$ Hamming cube under a projection to a Hamming cube on randomly chosen $k=25$ bits.}
\label{fig:distortion_h500}
\end{figure}

% 
% n <- 3000
% d <- 500
% k <- 25
% x <- matrix(rbinom(n*d, size=1, prob=0.5), ncol=d)
% r <- sample(d,k)
% y <- x[,r]
% dis <- as.matrix(dist(x, method = "manhattan"))
% disr <- as.matrix(dist(y, method = "manhattan"))*(d/k)
% proportion <- disr/dis
% postscript(file="distortion_h500.eps",
% horizontal=FALSE,onefile=FALSE,width=5,height=5)
% hist(proportion,breaks=40,col="orange",
% xlab="distortion",ylab="frequency of pairwise distances",main="")
% dev.off()

The error of $\pm \e$ is additive rather than multiplicative, so the random sampling of the coordinates is only appropriate for ANN search in the range on the order of $d/2$. The construction has to be generalized for all possible ranges $\ell=1,2,\ldots,d$. Such a generalization was developed in \cite{KOR:00}.

Projecting on a randomly sampled subset of $k$ coordinates of the Hamming cube essentially amounts to a linear transformation $x\mapsto xA$, where $A$ is a $d\times k$ matrix with i.i.d. Bernoulli entries assuming values $1$ and $0$ with probabilities $1/d$ and $1-1/d$, respectively. (The operations are carried $\mathrm{mod}\, 2$.) 
One of the key observations of \cite{KOR:00} --- in the form given to it in \cite{vempala}, 7.2 --- is that if the probability $1/d$ is replaced with $1/\ell$, then a random linear transformation $x\mapsto xA\,\mathrm{mod}\,2$, under a suitable normalization, preverves distances on the scale $\ell/2$, $\ell=1,2,\ldots,d$, to within an additive error $\ell\e$, and on a larger scale --- away from it. Since the new cube only contains $2^{O(\ve^{-2}\lg_2n)}$ points, a hash table storing nearest neighbours, together with the reduction map $f$, produces an indexing scheme for $\ell$-range search taking space polynomial in $n$ and answering $(1+\ve)$-approximate queries in time $O(\ve^{-2}\lg_2n)$. 

Another discovery of \cite{KOR:00} is that if on every scale $\ell$ one employs a sufficiently large series of independent projections onto $k$-cubes, then with high confidence one can assure that {\em every} ANN query --- as opposed to {\em most} ANN queries --- will be answered correctly. Finally, a separate indexing scheme is constructed for every range $\ell$. The overall space requirement is still polynomial in $n$, and the running time of the algorithm is $O(d\mathrm{poly}\,\log (dn))$.

\subsection{Random projections in the Euclidean space}

Let $\s^{d-1}$ denote the Euclidean sphere of unit radius in the space $\R^d$. The projection $\pi_1$ on the first coordinate is a $1$-Lipschitz function. For all pairs of points $x,y\in\s^{d-1}$, one has $\abs{\pi_1(x)-\pi_1(y)}\leq \norm{x-y}$, and for exactly one pair of antipodal points the equality is achieved. Now let $x,y\in\s^{d-1}$ be drawn at random. What is the expected value of the distortion of distances $\abs{\pi_1(x)-\pi_1(y)}/ \norm{x-y}$?

\begin{figure}[ht]
\centering
\epsfig{file=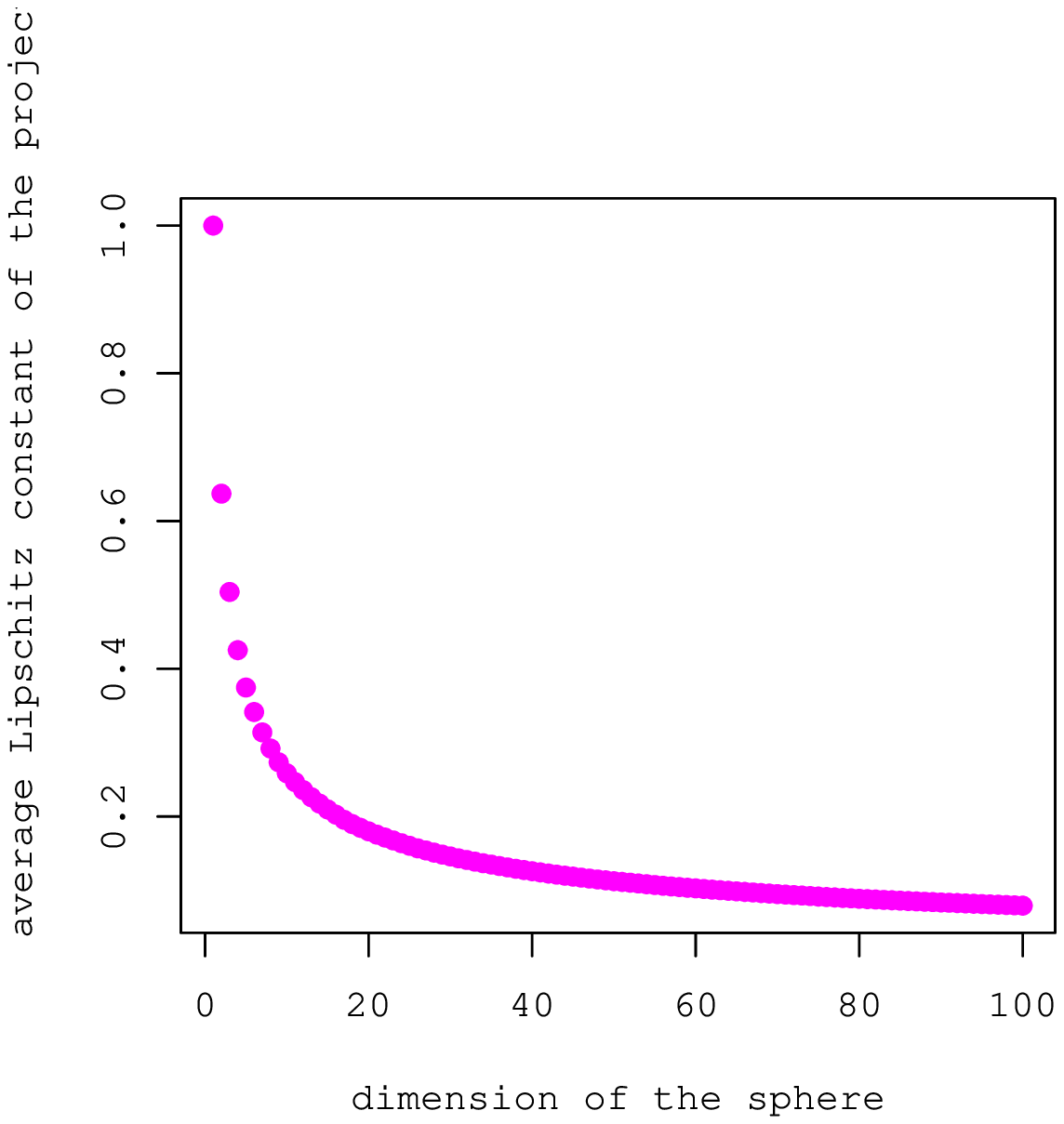, height=3.3in, width=3.3in} 
\caption{The expected distortion of one-dimensional projection of the $d$-dimensional sphere ${\mathbb S}^{d-1}$ over all pairs of points.}
\label{fig:lischitz100sph}
\end{figure}

% k <- 100
% L <- seq(from=0, to=0, length = k)
% xnorm <- matrix(rnorm(1000*d, mean=0, sd=1), ncol=d)
% for(d in 1:100){
% dis <- dist(xnorm[,1:d])
% disim <- dist(xnorm[,1])
% lip <- disim/dis
% L[d] <- mean(lip)}
% postscript(file="lipschitz100sph.eps",width=5,height=5,horizontal=FALSE,onefile=FALSE)
% plot(L,col="6",pch=19,main="",xlab="dimension of the sphere",ylab="average Lipschitz constant of the projection")
% dev.off()

Figure \ref{fig:lischitz100sph} shows that for a vast majority of pairs of points, the projection distorts distances by the factor $\Theta(1/\sqrt d)$. A geometric explanation, at least at an intuitive level, is simple. Two randomly chosen points on the high-dimensional sphere, because of concentration of measure, are at a distance $\approx \sqrt 2$ from each other. At the same time, half of the points of the sphere project on the interval of length $O(1/\sqrt d)$, and so are contained in the equatorial region (Figure \ref{fig:projection}).

\begin{figure}[ht]
\begin{center}
\scalebox{0.25}[0.25]{\includegraphics{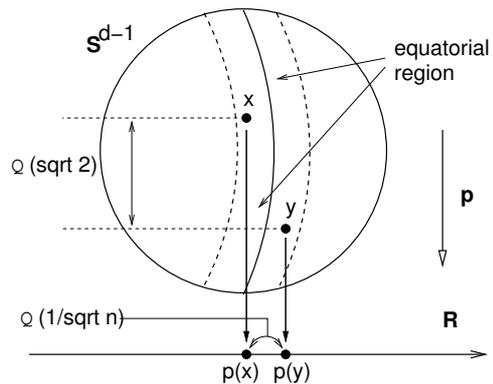}} 
\caption{\label{fig:projection}To the geometry of random projections.}
\end{center}
\end{figure} 

It follows that the expected absolute value of the norm of a projection of a given point $x$ in a random direction is of the order $\Theta(1/\sqrt d)$. Now let $X=\{x_1,\ldots,x_n\}$ be a finite subset of points of the sphere. Denote $Y$ the set of all vectors of the form $x_i-x_j$ whose length is normalized to one. Each $y\in Y$ can be identified with the function $\tilde y\colon z\mapsto \left\vert\langle y,z\rangle\right\vert$ on the unit sphere. If we now think of $\s^d$ as the domain (consisting of one-dimensional projections), then $Y$ plays the role of a finite {\em function class}. Just like for finite concept classes, the combinatorial dimension of $Y$ is of the order $O(\log n)$, and so, by VC theory, the empirical mean on a random sample of $\Theta(\e^{-1}\log n)$ directions will estimate the expectations of all $\tilde y$, $y\in Y$ to within a factor of $\e$ with high confidence. A small number of randomly chosen directions are likely to be nearly pairwise orthogonal because of concentration, so we can instead choose an orthogonal projection to a randomly chosen space of dimension $\Theta(\e^{-1}\lg n)$. Since the projection is a linear map, we get the same estimate, but with a {\em multiplicative error $\e$}, for all pairwise distances between the points of $X$. It remains to work out the meaning of the empirical mean in the above setting in order to obtain the following famous result.

\begin{theorem}[Johnson--Lindenstrauss lemma \cite{JL}] 
Let $\e\in (0,1/2)$ be a real number, and $X=\{x_1,x_2,\ldots,x_n\}$ be a set of $n$ points in $\R^n$. Let $k$ be an integer with $k\geq C\e^{-2}\log n$, where $C$ is a sufficiently large absolute constant. Then there is a mapping $f\colon\R^n\to\R^k$ such that
\[ 
(1-\e)\norm{f(x_i)-f(x_j)}\leq \norm{x_i-x_j} \]
\[\leq (1+\e)\norm{f(x_i)-f(x_j)} \]
for all $i,j=1,2,\ldots,n$. Moreover, as $f$, one can with high confidence choose a suitably renormalized random projection from $\R^n$ to a $k$-dimensional Euclidean subspace.
\end{theorem}

An even simpler proof using concentration can be found in \cite{matousek:2002}, Section 15.2, and an up-to-date survey of the lemma, in \cite{matousek:2008}.

The normalized projection is not quite as good as a genuine $1$-Lipschitz map, because the distortion of a distance can exceed one, and on rare occasions very considerably. Yet, as a reduction mapping for {\em approximate NN search,} the projection map is quite OK. And its
histogram is concentrated {\em no more}. This explains the efficiency of the random projection method for {\em approximate} NN search.
Combined with a suitable indexing scheme in a lower-dimensional space $\R^k$, or rather a collection of such schemes, the random projection method leads to an efficient indexing scheme for an $(1+\ve)$-approximate NN search (Indyk and Motwani \cite{IM}). 

The articles \cite{KOR:00} and \cite{IM} have appeared independently and at about the same time, and afterwards the dimensionality reduction methods have been shown \cite{AIP} to be near optimal in the cell probe model.

\section{Concluding remarks}
\subsection{Intrinsic dimensionality}
Merits of asymptotic analysis of indexing algorithms using artificial data\-sets sampled from theoretical high-dimensional distributions should be clear from \cite{navarro:2008}. At the same time, 
it is an often held belief that
the real data does not have very high intrinsic dimension. This corresponds to the existence of $1$-Lipschitz functions that are highly dissipating. Figure \ref{fig:density_best} shows the distance distribution to the points of the SISAP benchmark dataset of NASA images $X\subseteq\ell^2(20)$ of $40,149$ vectors in a $20$-dimensional Euclidean space \cite{BNC:03,sisap} from a highly dissipating pivot, selected from a gaussian cloud around $X$ with standard deviation on the order of the tolerance range $\ve=0.275$ retrieving on average $0.1\%$ 
of data. 
This has to be compared to Figure \ref{fig:hist_14gauss}.

\begin{figure}[ht]
\centering
\epsfig{file=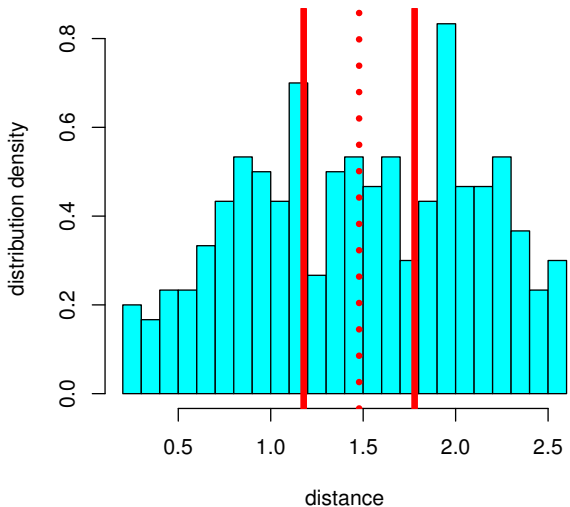, height=3.3in, width=3.3in}
\caption{Empirical density histogram of distances from a pivot having 
the highest found value of dissipation for the NASA dataset.  Vertical lines mark the mean $\pm$ tolerance range $\ve=0.275$. The $\epsilon$-dissipation ($0.747$) is the area outside of extreme lines.
\label{fig:density_best}}
\end{figure}

% > postscript(file="density_best.eps",horizontal=FALSE,onefile=FALSE,width=5,height=5)
% > hist(D[K[1],],breaks=20,col="5",freq=FALSE,main="pivot with the highest dissipation value",xlab="distance",ylab="distribution density")
% > abline(v=(m+0.3),col="red",lwd=5)
% > abline(v=m,col="red",lwd=5,lty=3)
% > abline(v=(m-0.3),col="red",lwd=5)
% > dev.off()

Of a great variety of approaches to intrinsic dimension \cite{clarkson:06}, at least two specifically measure the amount of concentration in data. The first one is the intrinsic dimension by Ch\'avez {\em et al.} \cite{chavez:01}
\begin{equation}
\dim_{dist}(X)=\frac{1}{2\var(d)}.
\label{eq:int}
\end{equation}
The second is the concentration dimension, studied within an axiomatic approach of \cite{pestov:07,pestov:08}:
\begin{equation}\label{eq:concdim}
\dim_{\alpha}(X)=\frac{1}{\left[2\int_0^1 \alpha_X(\e)~d\e\right]^2}.\end{equation}
(In both cases we assume that ${\mathrm{CharSize}}(X)=1$.)
The value (\ref{eq:concdim}) is convenient for asymptotic analysis in the spirit of this paper, but is nearly impossible to estimate for a given dataset.
On the other hand, (\ref{eq:int}) is readily calculated by sampling (e.g. $\dim_{dist}(X)=5.18$ for NASA images) and forms a good statistical estimator for the dimension of the hypothetical underlying measure $\mu$ in the most (only?) interesting case where metric balls have low VC dimension. The shortcoming of (\ref{eq:int}) is that the parameter estimates the concentration/dissipation behaviour of a {\em typical} pivot distance function, while it is a few most dissipating pivots that really matter for indexing.
% library(Matrix)
% x <- read.table("/Users/vova/Data/metricSpaces/lib/vectors/nasa/x.txt", 
% +   header=TRUE, sep="", na.strings="NA", dec=".", strip.white=TRUE)
 % 
 % x <- as.matrix(x)
 % library(Rcmdr)
% d <- seq(from=0, to=0, length=200)
% for( i in 1:200){
% y<- x[sample(40149,5000),]
% distances <- dist(y)
% intdim <- (mean(distances)^2)/(2*var(distances))
% d[i] <- intdim
% }
One may envisage the emergence of further concepts of intrinsic dimension in the same spirit, such as the {\em local dimension} of Ollivier \cite{ollivier:07}, Definition 3.

\subsection{Black box search model and Urysohn space}

The {\em black box model} of similarity search was studied by Krauthgamer and Lee \cite{KL:05}. Given a metric space (instance) $(X,d)$, a query is a one-point metric space extension $X\cup\{q\}$, where the distances $d(q,x)$, $x\in X$ are accessible via the distance oracle. Each $d(q,x)$ can be evaluated in constant (unit) time. A preprocessing phase is allowed, under the condition that an indexing scheme occupies ${\mathrm{poly}}\,(n)$ space. The efficiency of an algorithm for (exact or approximate) similarity search is estimated as a number of calls to the distance oracle necessary to answer a query.

This is a ``black box model'' in the sense that, formally speaking, there is no obvious domain (though we will see shortly that the domain is a well-defined separable metric case, and the setting is, in fact, classical). A remarkable feature of the model is that the problem of characterizing workloads admitting approximate NN queries in terms of an intrinsic dimension parameter receives a complete answer.

Recall that the {\em Assouad} (or {\em doubling}) {\em dimension} of a metric space $(X,d)$ is the minimum value $\rho\geq 0$ such that every set $A$ in $X$ can be covered by $2^{\rho}$ balls of half the diameter of $A$. (The {\em diameter}  of a set $A$ is the supremum of $d(x,y)$, $x,y\in A$.) Denote this parameter by $\dim_{dbl}(X)$.

\begin{theorem}[Krauthgamer and Lee \cite{KL:05}]
A metric space $(X,d)$ admits an algorithm requirying ${\mathrm{poly}}\,(n)$ space and taking ${\mathrm{polylog}}\,(n)$ time to answer a $(1+\e)$-approximate nearest neighbour query, where $\e<2/5$, if and only if \[\dim_{dbl}(X,d)=O(\log\log n).\]
\end{theorem}
 
Here we will show that, on the contrary, an {\em exact} NN search in this context exhibits the curse of dimensionality even if the metric space $(X,d)$ is contained in the unit interval $[0,1]$ with the usual distance. 
With this purpose, we first convert the black box model into a conventional setting of searching in a metric domain.

The {\em universal Urysohn metric space,} $\Ur$, \cite{gromov:99,melleray} is a complete separable metric space uniquely defined by the {\em one-point extension property:} suppose $X$ is a finite subset of $\Ur$ and $q$ a one-point metric space extension of $X$. Then $\Ur$ contains a point $q^\prime$ so that the distances from $q$ and from $q^\prime$ to any point $x\in X$ are the same. 

\begin{figure}[ht]
\begin{center}
\scalebox{0.25}[0.25]{\includegraphics{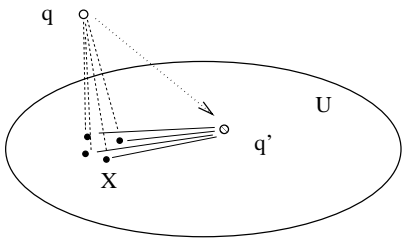}} 
\caption{\label{fig:onept_ext} One-point extension property.}
\end{center}
\end{figure}

An equivalent definition is that if $X\subseteq \Ur$ is finite and $f\colon X\to\R$ satisfes
\begin{equation}
\vert f(x)-f(y)\vert\leq d_X(x,y)\leq f(x) + f(y)
\label{eq:flood}
\end{equation}
for all $x,y\in X$, then there is $q\in\Ur$ with $f(x)=d(q,x)$ for all $x\in X$. (The functions satisfying (\ref{eq:flood}) are called {\em Kat\v etov functions}.)

This remarkable object has recently received plenty of attention in metric geometry. It is a {\em random,} or {\em generic,} metric space, in a sense that by equipping the integers with a randomly chosen metric $\rho$ and taking a completion, one obtains $\Ur$ almost surely \cite{vershik}. The space $\Ur$ contains an isometric copy of every separable metric space $\Omega$. For this reason, one can use $\Ur$ as a ``universal domain,'' and
the black-box model can be restated as a classical similarity search problem in the domain $\Omega=\Ur$.
%and requiring that every function evaluated during the execution of the search algorithm depend on distance functions to pivots, $d(p_i,-)$, $p_i\in\Ur$. 
 
\begin{theorem}
Let $X$ be a finite metric space. Denote $n=\abs X$. Then any deterministic algorithm for exact similarity search in $X$ within the black box model will take the worst case time $n$. 
\end{theorem}

The result is true for simple information-theoretic reasons. We will produce for every $k<n$ a query $q^{\prime\prime}$ with a uniquely defined nearest neighbour in $X$ which cannot be answered in time $k$.

Without loss in generality, we can assume that $\diam(X)=1$. Let initially $q$ be a query having the property that $d(q,x)=1$ for all $x\in X$. 
Suppose that the algorithm has made $k<n$ calls to the distance oracle. Denote $x_1,x_2,\ldots,x_k\in X$ the points whose distance to $q$ has been accessed.  % Let $i={\mathrm{argmin}}_{i=1}^kd(x_i,q)$, and denote $\e=d(x_i,q)$. 
Since $d(q,x_i)=1$ for all $i\leq k$, the algorithm clearly cannot halt at this stage.
Let $Q$ be the set of all $q^\prime\in \Ur$ with $1=d(q^\prime,x_i)$ for all $i=1,2,\ldots,k$. Since the algorithm is deterministic, we can replace $q$ with any $q^\prime\in Q$, and the sequence of executed calls to the oracle up until the step $k$ will be the same. 

Now denote $Y=\{x_1,x_2,\ldots,x_k\}$ and fix an $x_0\in X\setminus Y$. The function
\[f(x) = \max\{1-d(x,Y),d(x_0,Y)-d(x,x_0) \}\]
is Kat\v etov, and thus it is the distance function from some $q^{\prime\prime}$. Clearly, $q^{\prime\prime}\in Q$, and $q^{\prime\prime}$ admits a unique nearest neighbour in $X$, namely $x_0$.
Thus, the search cannot be concluded in $k$ steps even if it started with the well-defined query $q^{\prime\prime}$. \qed
 
If one requires the queries to follow the same underlying distribution as datapoints, the problem becomes more subtle, and we do not know the answer.

\subsection{Indexing via Delaunay graph}
Here is an example of an indexing scheme for exact similarity search which is still ``distance-based'' but of a rather different type from either pivots or metric trees.

The {\em Voronoi cell} $V(x)$ of a datapoint $x\in X$ in a metric domain $\Omega$ consists of all points $q\in\Omega$ having $x$ as the nearest neighbour. The {\em Delaunay graph} has $X$ as the set of vertices, with $x,y$ being adjacent if their Voronoi cells intersect. Suppose the domain has the property that every two points $x,y\in\Omega$ can be joined by a shortest geodesic path, not necessarily unique. (All the domains previously considered in this article are such, including even the Urysohn space.) Then for any $q\in \Omega$ and $x\in X$, either $x$ is the nearest neighbour to $q$, or else one of the datapoints $y$ Delaunay-adjacent to $x$ is strictly closer to $q$ than $x$ is. (Proof: start moving along a shortest geodesic from $x$ towards $q$, cf. Figure \ref{fig:delaunay8}, and use the triangle inequality.)

\begin{figure}[ht]
\centering
\scalebox{0.25}[0.25]{\includegraphics{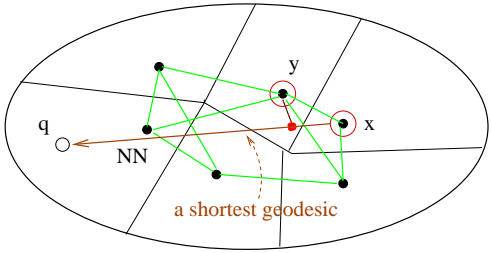}} 
\caption{NN search using Delaunay graph.}
\label{fig:delaunay8}
\end{figure}

This observation turns the Delaunay graph of $X$ in $\Omega$ into an indexing scheme for exact nearest neighbour search. Denote $S_x$ the list of points adjacent to each $x\in X$. Given a query $q$, start with an arbitrary $x_0\in X$, and find
\[x_1 =\arg\min_{y\in S_x}d(q,y).\]
If $x_1\neq x_0$, move to $x_1$ and repeat the procedure. Once $x_{i+1}=x_i$, the algorithm halts and returns $x_i$. 
This algorithm, already mentioned in \cite{clarkson:94}, was studied for general metric spaces by Navarro \cite{navarro:2002}. See also \cite{samet}, 4.1.6.

In order for the algorithm to be efficient, the average vertex degree of the Delaunay graph has to be small.
Navarro had observed ({\em loc. cit.,} Theorem 1) that this is not the case in general metric spaces. Specifically, he proved that for every two elements $a,b\in X$ there exists a finite metric space $Y=Y_{a,b}$ containing $X$ as a subspace in which $a,b\in X$ are connected in the Delaunay graph of $X$. 
The result by Navarro translates immediately into:

\begin{theorem}
Let $X$ be a finite metric subspace of the universal Urysohn space $\Ur$. Then every two elements $a,b\in X$ are adjacent in the Delaunay graph of $X$ in $\Ur$.
\end{theorem}

In fact, the same remains true in less exotic situations, as one can deduce from Proposition \ref{p:simplex} that if $\Omega$ be either $\R^n$, or the sphere $\s^n$, or the Hamming cube, then under the assumptions of Subs. \ref{ss:assumptions} the Delaunay graph of $X$ is, with high confidence, a complete graph on $n$ vertices. 

Thus, the indexing scheme in question still suffers from the curse of dimensionality because of concentration of measure considerations, but the argument seems to be of a different nature from that either for pivots or for trees. What would a common proof for all three types of schemes look like?
This highlights the difficulty of obtaining in a uniform way lower bounds for all possible ``distance-based'' indexing schemes (after they are formalized in a suitable way), not to mention an even more general setting of the cell probe model for all possible indexing schemes.

This having said, for real data the complexity of the Delaunay graph is lower than in an artificial asymptotic setting, and Voronoi diagrams are being successfully used for data mining algorithms in high dimensions, cf. \cite{TM:09}. 

In fact, it would be interesting to investigate the performance of the spatial approximation algorithm in hyperbolic metric spaces. Recall that a metric space $X$ in which every two points $x,y$ can be joined by a geodesic segment $[x,y]$ is {\em hyperbolic} (in the sense of Rips) \cite{vaisala} if there exists a $\delta>0$ so that every geodesic triangle is {\em $\delta$-thin:} each side $[x,y]$ is contained in the $\delta$-neighbourhood of the two other sides, $[x,z]$ and $[y,z]$ (Figure \ref{fig:geodesic}).
\begin{figure}[ht]
\begin{center}
\scalebox{0.25}[0.25]{\includegraphics{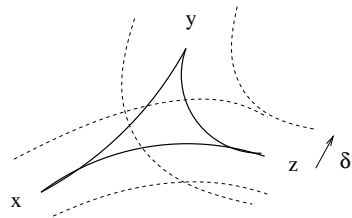}} 
\caption{\label{fig:geodesic} A $\delta$-thin geodesic triangle.}
\end{center}
\end{figure}

Alain Connes has conjectured in \cite{CC}, pp. 138--141 that a long-term human memory is organized as a hyperbolic simplicial complex, where a search is performed in a manner similar to the above.
\appendix

\section{Proof of the empty space paradox (Theorem \ref{th:empty} and Proposition \ref{p:simplex})}
Without loss in generality, normalize the observable diameter of $\Omega$ to one.
Let $\omega\in\Omega$. The distance function $\rho(\omega,-)$ is $1$-Lipschitz and so concentrates around its median value, $R(\omega)$. The resulting function $R\colon\Omega\to\R$, $\omega\mapsto R(\omega)$ is also $1$-Lipschitz, and concentrates around its median, $R_M$. It is easy to check that under our assumptions, the difference between the mean and the median of every $1$-Lipschitz function $f$ on $\Omega$ converges to zero as $O(\sqrt d)$ (uniformly in $f$). Thus, without a loss in generality, we can assume that, with high confidence, $R_M\to 1$ as $d\to\infty$. Notice that the above argument concerns the {\em domain} and not a particular {\em dataset}.

To prove Proposition \ref{p:simplex}, fix $\e>0$ and sample an instance of data, $X$. With confidence $1-n\exp(-O(d)\e^2)$, one has $\abs{R(x)-1}<\e/2$ for all $x\in X$. Moreover, since the datapoints are sampled in an i.i.d. fashion, by the union bound one has with confidence $1-n^2\exp(-O(d)\e^2)$ that $\abs{\rho(x,y)-R(x)}<\e/2$ for every pair $x,y\in X$. Since $n=\abs X$ is subexponential in $d$, the statement follows.

To prove Theorem \ref{th:empty}, again fix $\e>0$.
%When $X$ is sampled, with extremely high confidence at least one point in $x$ satisfies $R(x)>1-\e$. Since the ball of radius $R(x)$ around $x$ has measure $\geq 1/2$, the same holds for the $1$-neighbourhood of $X$.
Denote $\ve_M$ the median value of the function $\ve_{NN}$. 
%We conclude: with high confidence, the value of the random variable $\ve_M$ exceeds $1-\gamma$. 
Suppose $\liminf_{d\to\infty}\e_M<1$. Proceed to a subsequence of domains and find $\gamma>0$ with $\ve_M\leq R_M-\gamma$ for all $d$. The probability that $R(\omega)$ deviates from $R_M$ by more than $\gamma/2$ is exponentially small in $d$. Since $n=\abs X$ only grows subexponentially in $d$, with confidence $1-\exp(-O(\e^2 d))$ one has for every $x\in X$:
\[R(x)- \ve_M \geq \frac\gamma 2.\]
Now we use a technical observation from \cite{GrM:83}: if $A\subseteq\Omega$ is such that $\mu(A)>\alpha_{\Omega}(\gamma)$ for some $\gamma>0$, then $\mu(A_{\gamma})>1/2$. It follows that
\[\mu\left(B_{\ve_M}(x) \right)\leq \alpha(\gamma/2) = \exp(-O(\e^2 d)),\]
and therefore
\[\mu\left(X_{\ve_M} \right)\leq n \exp(-O(\e^2 d)) = \exp(-O(\e^2 d)),\]
which contradicts the definition of $\ve_M$. This implies: $\liminf_{d\to\infty}\e_M\geq 1$.

To establish the converse inequality $\limsup_{d\to\infty}\e_M\leq 1$, recall that a ball of radius $R(\omega)$ centred at $\omega$ has measure $\geq 1/2$, and so we have an obvious estimate $\e_M\leq \min_{x\in X}R(x)$. The rest follows from concentration of the function $R$ around one. 

\section{\label{s:distortion}Distortion of Lipschitz embeddings $\ell^2(d)\hookrightarrow \ell^{\infty}(d+k)$}

\noindent
{\bf Lemma \ref{eq:c}} {\em
Fix $k$.
Let $c>0$ be a constant having the property that for every $d$ and each bounded subset $X$ of $\ell^2(d)$ there exists a 1-Lipschitz function $f\colon X\to\ell^\infty(d+k)$ having distortion $c$: for all $x,y\in X$,
\begin{eqnarray}
\label{eq:c}
\norm{f(x)-f(y)}_{\infty}&\leq&\norm{x-y}_2 \\
&\leq& c\norm{f(x)-f(y)}_{\infty}.\nonumber
\end{eqnarray}
Then $c=\Omega(\sqrt {d+k}/\sqrt{k})$, that is, $\Omega(\sqrt d)$ with a constant depending on $k$.}
\\[2mm]

The proof consists of a series of statements.

1. {\em There exists a 1-Lipschitz function $f\colon \ell^2(d)\to\ell^\infty(d+k)$ with the property (\ref{eq:c}).}
 
For every $n\in\N$, choose a function $f_n$ from the closed $n$-ball $B_n(0)$ in $\ell^2(d)$ to the $n$-ball in $\ell^\infty(d+k)$ with distortion $c$. The Banach space ultrapowers of both participating spaces formed with regard to a non-principal ultrafilter on the integers (see e.g. page 55 in \cite{JL2}) are isometric, respectively, to $\ell^2(d)$ and $\ell^\infty(d+k)$, because the spaces in question are finite-dimensional. The family of $1$-Lipschitz functions $(f_n)$ determines in a standard way a $1$-Lipschitz function, $f$, from $\ell^2(d)$ to $\ell^\infty(d+2)$, with the property (\ref{eq:c}) being preserved. 

2. {\em There exists a linear function $f$ with the property (\ref{eq:c}).}

Choose $f$ as in 1.
According to the Rademacher theorem (cf. a discussion and references on p. 42 in \cite{JL2}), $f$ is differentiable almost everywhere with regard to the Lebesgue measure. The differential of $f$ at any point, which we denote $T$, is a linear operator of norm one having property (\ref{eq:c}). In particular it is injective (though of course not onto), and the inverse has norm $\leq c$.

Recall that the {\em multiplicative Banach--Mazur distance} between two normed spaces $E$ and $F$ of the same dimension is the infimum of all numbers $\norm T\cdot \norm T^{-1}$, where $T$ ranges over all isomorphisms between $E$ and $F$. (See \cite{JL2}, p. 3, and \cite{GiM}, 7.2).
From the previous observation, we conclude:

3. {\em The Banach--Mazur distance between $\ell^2(d)$ and some $d$-dimensional subspace of $\ell^{\infty}(d+k)$ is $\leq c$.}

4. {\em The Banach--Mazur distance between $\ell^2(d+k)$ and $\ell^{\infty}(d+k)$ is $O(\sqrt k c)$.}

There is a projection $p$ from $\ell^{\infty}(d+k)$ having $T(\ell^2(d))$ as its kernel and such that $\norm p\leq \sqrt k$ and $\norm{1-p}\leq \sqrt k$ (combine \cite{davis}, Corollary on page 209, with a classical result of Kadec and Snobar on projection constants, cf. \cite{JL2}, p. 71). The Banach-Mazur distance between $\ell^2(k)$ and any other $k$-dimensional normed space, including the kernel of $p$, is $O(\sqrt k)$. Choose an isomorphism $S$ realizing this distance, then it is easy to verify that $T\oplus S$ realizes the distance $O(\sqrt k c)$ between $\ell^2(d+k)$ and $\ell^{\infty}(d+k)$.

Finally, the Banach--Mazur distance between $\ell^2(d+k)$ and $\ell^{\infty}(d+k)$ is $\sqrt{d+k}$ (cf. \cite{GiM}, p. 766).

\bibliographystyle{elsarticle-num}

\end{document}